\documentclass[12pt]{article}
\usepackage{amsmath}
\usepackage{amsthm}
\usepackage{graphicx}
\usepackage{enumerate}
\usepackage[numbers]{natbib}
\usepackage{bibunits}
\usepackage{url} % not crucial - just used below for the URL 
\newcommand{\blind}{1}

\usepackage{amssymb }
\usepackage{enumitem}
\usepackage{algorithm2e}
\usepackage{subcaption}
\usepackage{chapterbib}
\usepackage[colorlinks=true, allcolors=black]{hyperref}
\usepackage{titletoc}

\DeclareUnicodeCharacter{008D}{}
\DeclareUnicodeCharacter{0080}{}
\DeclareUnicodeCharacter{0093}{}

\newtheorem{theorem}{Theorem}[section]

\newtheorem{definition}[theorem]{Definition}

\usepackage{setspace}
\usepackage{booktabs} % For better horizontal rules
\usepackage{multirow}
\usepackage{caption} % For customizing captions
\usepackage{siunitx} % For improved decimal alignment
\usepackage{threeparttable}
\usepackage{bbm}
\usepackage{bm}    
\usepackage{xcolor}
\usepackage{sectsty}
\usepackage{tikz}
\usetikzlibrary{arrows.meta,positioning,calc}

\begin{document}

%\begin
%%%%%%%%%%%%%%%%%%%%%%%%%%\begin{bibunit}[unsrt]
%\bibliographystyle{natbib}

\def\spacingset#1{\renewcommand{\baselinestretch}%
{#1}\small\normalsize} \spacingset{1}

%%%%%%%%%%%%%%%%%%%%%%%%%%%%%%%%%%%%%%%%%%%%%%%%%%%%%%%%%%%%%%%%%%%%%%%%%%%%%%

\if1\blind
{
  \title{\sf Deep Computerized Adaptive Testing}
   \author{Jiguang Li\footnote{Jiguang Li is a $4^{th}$-year doctoral student in Econometrics and Statistics at the Booth School of Business of the University of Chicago},\,\, Robert Gibbons\footnote{Robert Gibbons is the Blum-Reise Professor of Statistics at the University of Chicago}\,\, and  Veronika Ro\v{c}kov\'{a}\footnote{Veronika Ro\v{c}kov\'{a} is the Bruce Lindsay Professor of Econometrics and Statistics in the Wallman Society of Fellows}}
  \maketitle
} \fi

\if0\blind
{
  %\bigskip
  % \bigskip
  % \bigskip
   \title{\sf Deep Computerized Adaptive Testing}
  \maketitle
  \medskip
} \fi

%\bigskip
\begin{abstract}
Computerized adaptive tests (CATs) play a crucial role in educational assessment and diagnostic screening in behavioral health. Unlike traditional linear tests that administer a fixed set of pre-assembled items, CATs adaptively tailor the test to an examinee’s latent trait level by selecting a smaller subset of items based on their previous responses. Existing CAT applications predominantly rely on item response theory (IRT) models with a single latent variable, a choice driven by both conceptual simplicity and computational feasibility. However, many real-world item response datasets exhibit complex, multi-factor structures, limiting the applicability of CATs in broader settings. In this work, we develop a novel CAT system that incorporates multivariate latent traits, building on recent advances in Bayesian sparse multivariate IRT \cite{doi:10.1080/01621459.2025.2476786}. Our approach leverages direct sampling from the latent factor posterior distributions, significantly accelerating existing information-theoretic item selection criteria by eliminating the need for computationally intensive Markov Chain Monte Carlo (MCMC) simulations. Recognizing the potential suboptimality of existing item selection rules which are often based on myopic one-step-lookahead optimization of some information-theoretic criterion, we propose a double deep Q-learning algorithm to learn an optimal item selection policy. Through simulation and real-data studies, we demonstrate that our approach not only accelerates existing item selection methods but also highlights the potential of reinforcement learning in CATs. Notably, our Q-learning-based strategy consistently achieves the fastest posterior variance reduction, leading to earlier test termination. In finite-horizon settings, it also yields final posterior distributions that more closely approximate the oracle posterior distribution—the posterior that would be obtained if test takers had responded to all items in the item bank.
\end{abstract}

\noindent%
\vfill

\newpage
\spacingset{1.8}

\section{Introduction}
Multidimensional Computerized Adaptive Testing (MCAT) has revolutionized the field of educational and psychological assessments by dynamically selecting tailored items from a large test pool, thereby enhancing the efficiency and precision of latent ability estimates \citep{van2010elements}. Powered by Multidimensional Item Response Theory (MIRT) \citep{bock2021item}, MCAT leverages multidimensional statistical inference to evaluate respondents' multidimensional latent traits, and allows for more comprehensive and efficient assessments compared to unidimensional approaches. MCAT’s adaptability and accuracy are also particularly crucial in high-stakes diagnostic assessments such as in clinical psychology and psychiatry, where it can be substituted for in-person assessments by clinical professionals especially in areas with limited medical resources \citep{gibbons2016computerized}.  

A substantial body of research in MCAT has focused on item selection strategies derived from experimental design principles. One prominent strategy is to select items that maximize the determinant of the Fisher information matrix evaluated at the current estimates of the latent traits \citep{segall1996multidimensional, segall2000principles}, known as the D-optimality criterion. An alternative, the A-optimality criterion, aims to minimize the trace of the asymptotic covariance matrix, thereby reducing overall estimation variance \citep{a584ed95-55e4-36a9-8d67-ff8dd3f06250}. In the absence of nuisance latent abilities, both A-optimality and D-optimality have demonstrated superior accuracy relative to other common experimental design criteria \citep{mulder2009multidimensional}.

Another prominent line of MCAT item selection algorithms leverages Kullback-Leibler (KL) information, often within a Bayesian framework. A common approach is to select items that produce response distributions at the true latent trait value, $\bm{\theta}_0$, that differs maximally from the response distributions generated at the other value of $\theta$ \citep{Chang1996AGI, pekl}. Moving beyond response distributions alone, some researchers propose maximizing the KL divergence between the current posterior distribution and the posterior distribution at the next selection step, thereby enhancing adaptation through updated trait estimates \citep{kl-is}. Another promising strategy is the mutual information criterion, which aims to maximize entropy reduction of the current posterior distribution, encouraging more and more accurate posterior estimates \citep{mi_selection}. In particular, \cite{selection_survey} demonstrates both theoretical and empirical advantages of the Bayesian mutual information item selection rule over common experimental criteria such as D-optimality. More detailed theoretical comparison of KL information and Fisher information criteria is presented by \cite{Wang_Chang_Boughton_2011}. 

Although numerous effective item selection rules have been proposed in the MCAT literature, they all rely on one-step lookahead optimization of an information-theoretic criterion. Despite their ease of implementation and attractive theoretical properties, these selection rules are inherently myopic: they select items based solely on immediate information gain, ignoring how current choices influence future decisions, which can lead to suboptimal policies.
For example, existing methods tend to favor items with high loading parameters \citep{Chang2015}. However, CAT researchers also recommend reserving such items for the later stages of testing to improve efficiency \citep{doi:10.1177/01466219922031338}. Integrating heuristic guidance into existing selection rules remains challenging.

Addressing these limitations, we propose a novel deep CAT system that integrates a flexible Bayesian MIRT model with a non-myopic online item selection policy, guided by reinforcement learning (RL) principles \citep{Sutton1998}. Leveraging recent advancements in Bayesian sparse MIRT \citep{doi:10.1080/01621459.2025.2476786}, our framework seamlessly accommodates multiple latent factors with complex loading structures, while maintaining scalability in both the number of items and factors. To learn the optimal item selection policy that prioritizes the assessment of target factors, we draw on contemporary RL methodologies and introduce a general double deep Q-learning algorithm \citep{Mnih2015,van_Hasselt_Guez_Silver_2016}. This algorithm efficiently trains a Q-network offline using only item parameter estimates; the learned network can then be deployed online to select optimal items based on the current multivariate latent factor posterior distribution. When the test terminates, our framework robustly characterizes the full latent factor posterior distributions rather than providing only a point estimate.

A primary contribution of our work is to show how identification of the latent-factor posterior distribution leads to substantial computational gains in online item selection. Given that such posterior distribution is deemed to be non-Gaussian and unknown, traditional Bayesian methods typically rely on expensive MCMC simulations \citep{BeguinGlas2001} and combined with additional data augmentation to handle categorical likelihood \citep{f3928584-48a1-3ab1-b404-2ec9ffdf51bb, Polson01122013}. Our approach achieves substantial acceleration by directly sampling latent factor posterior distributions \citep{doi:10.1080/01621459.2025.2476786, Durante_2019, Botev_2016}. Notably, this improvement not only increases the efficiency of existing Bayesian item-selection procedures but also provides a computational foundation for training our proposed Q-learning algorithm through rapid, large-scale simulations of testing sessions.

Another critical advancement in our work is the integration of CAT within a reinforcement learning framework. This approach addresses the practical need to prioritize accurate estimation and to overcome known limitations of greedy item selection methods in sequential decision-making \citep{10.5555/560669}. Building on the Bayesian MIRT foundation, our neural-network architecture incorporates the identified posterior parameters as state variables, allowing the model to learn optimal item-selection policies through a large amount of testing simulations. This formulation bridges the two paradigms: the Bayesian component provides statistically grounded representations of examinee uncertainty, while the reinforcement-learning component leverages these representations to optimize sequential decisions. The trained neural network is deployable on standard laptops without GPU acceleration, making it suitable for online adaptive testing applications. The sequential nature of CAT aligns naturally with deep Q-learning methods, which have demonstrated remarkable success across diverse application domains \citep{Silver2016,Kalashnikov2018QTOptSD}. Notably, RL has been successfully employed in educational measurement settings to design personalized learning plans \citep{doi:10.3102/10769986221129847, doi:10.1177/0146621619858674}. 

Finally, our work aligns with the emerging research trend of framing traditional statistical sequential decision-making problems as optimal policy learning tasks \citep{modern_bad_review}. This perspective has appeared in Bayesian adaptive design \citep{bad_review,bad,Foster2021DeepAD} and Bayesian optimization \citep{srinivas2010gaussian,NIPS2016_5ea1649a}, where recent methods aim to move beyond one-step criteria toward policy learning that account for long-term consequences. Our contribution fits within this broader trend by providing a principled approach to cognitive and behavioral assessment.

The remainder of the paper is organized as follows. Section \ref{sec:motivation} motivates our deep CAT framework using a cognitively complex item bank from a recent clinical study \citep{cat-cog}. Section \ref{sec:problem-formulation} reviews existing information-theoretic item selection methods and outlines the necessity to reformulate CAT as a reinforcement learning task. Section \ref{sec:acceleration} introduces a general Bayesian framework that accelerates existing CAT item-selection rules and serves as the foundation for our Q-learning algorithm. Section \ref{sec:Q-rl} details our reinforcement learning approach, including the neural network architecture and double Q-learning algorithm. Finally, Sections \ref{sec:experiments} and \ref{sec:cat-cog} evaluate our method through both simulations and real data experiments. The implementations of the baseline Bayesian CAT algorithms, our Q-learning algorithm, and the driver codes to reproduce all the experiments in the paper can be found in \url{https://github.com/JiguangLi/deep_CAT}.

\section{Adaptive High-Dimensional Cognitive Assessment} \label{sec:motivation}

Our proposed deep CAT system is motivated by the growing need for adaptive cognitive assessment in high-dimensional latent spaces. Cognitive impairment, particularly Alzheimer’s dementia (AD), is a major public health challenge, affecting 6.9 million individuals in the U.S. in 2024 \citep{alzheimers2024facts}. Early detection of mild cognitive impairment (MCI), a precursor to AD, is crucial for slowing disease progression and improving patient outcomes \citep{10.1001/jamanetworkopen.2018.1726}. However, traditional neuropsychological assessments are costly, time-consuming, and impractical for frequent use, underscoring the need for more efficient and scalable assessment methods.

Recently, pCAT-COG, the first computerized adaptive test (CAT) item bank based on MIRT for cognitive assessment, demonstrated its potential as an alternative to clinician-administered evaluations \citep{cat-cog}. The data were collected from $730$ participants. After careful item calibration and model selection, the final item bank consisted of $J=57$ items covering five cognitive subdomains: episodic memory, working memory, executive function, semantic memory, and processing speed.  Since each item comprised three related tasks, we used a binary score, where $1$ indicates correct answers on all tasks for our analysis. 

Following \cite{cat-cog}, we fit a six-factor bifactor model \citep{bifactor} to the $730$ by $57$ binary response matrix, with one primary factor representing the global cognitive ability and five secondary factors. This yields a $57$ by $6$ factor loading matrix, visualized in Figure \ref{fig:cat-cog-loadings}, where rows correspond to pCAT-COG items and columns represent distinct factors. 

\begin{figure}[t]
    \centering
    \includegraphics[width=0.65\textwidth, height=0.5\textwidth]{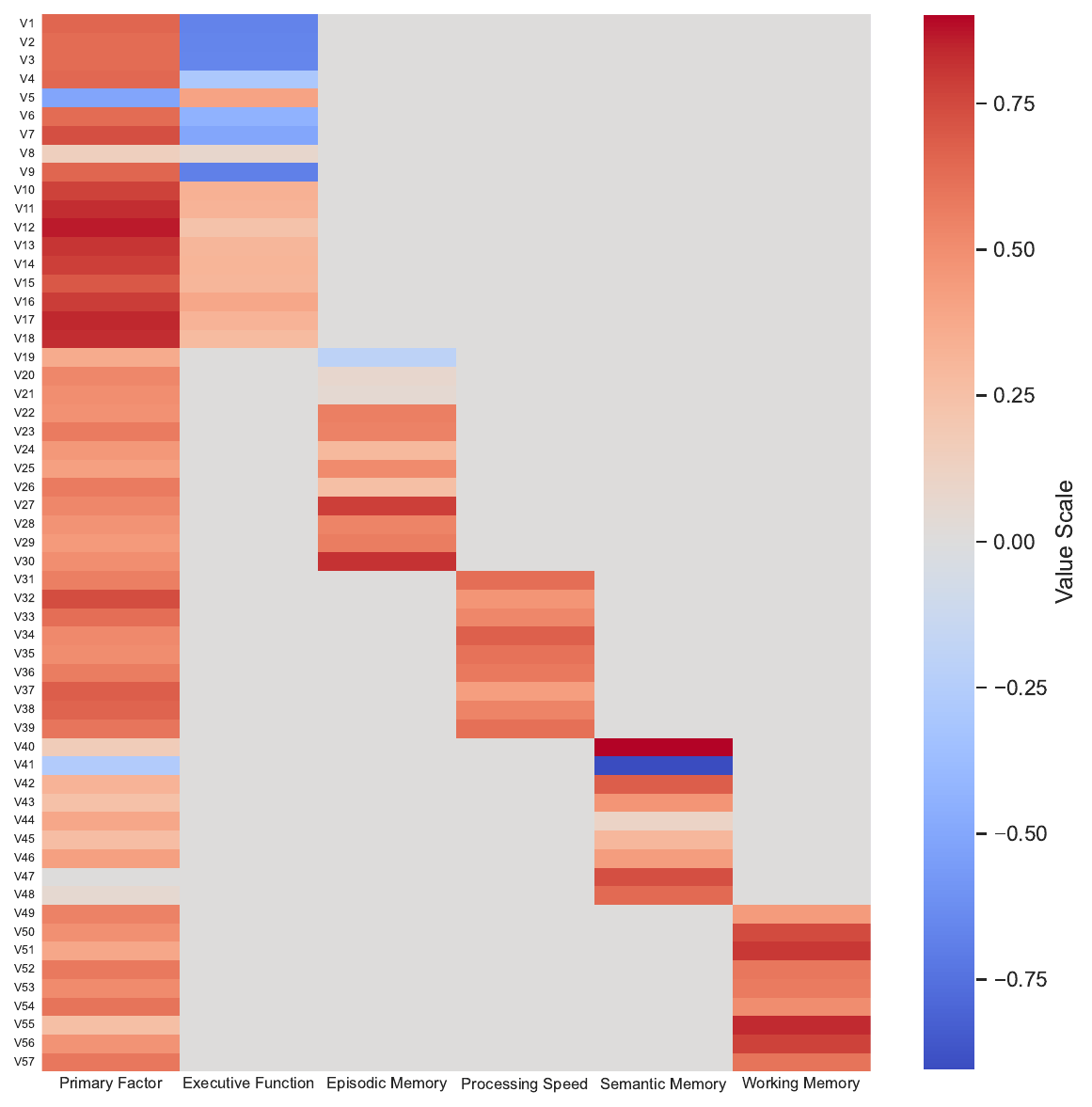}
    \caption{Estimated Bifactor Factor Loading Matrix for pCAT-COG}
    \label{fig:cat-cog-loadings}
\end{figure}

Since items in pCAT-COG are primarily designed to measure the general cognitive factor and only partially capture subdomain information, our goal is to develop an item selection strategy that efficiently estimates the primary factor (first column) using as few items as possible while maintaining robustness to subdomain influences. Additionally, the selection algorithm must be computationally efficient to navigate a six-dimensional latent space in real time for seamless interactive testing. Online item selection is critical, as learning an optimal sequence from $57$ items offline presents an intractable combinatorial problem, even in a binary response setting.

Adaptive cognitive assessment is essential for cognitive assessment, as designing items is costly and administering all items is time-consuming. Our deep CAT system requires on average $11.2$ items to reduce the posterior variance of the primary factor from $1$ to $0.16$ (posterior s.d. $= 0.4$), whereas the next-best mutual information method requires $12.6$ items to achieve the same precision. The item bank is currently expanding to $500$ items with a more nationally representative sample. Success on this prototype dataset paves the way for broader deployment in clinical research.

\section{From One-Step Optimization to Policy Learning} \label{sec:problem-formulation}
We formally define the problem of designing a CAT system from a reinforcement learning (RL) perspective. Typical CAT systems consist of two components:

\begin{itemize}
    \item \textbf{Offline Calibration}: An MIRT model is fitted to a calibrated dataset $\bm{Y} \in \mathbbm R^{N \times J}$ to estimate the item characteristic parameters, where $N$ represents the number of examinees, and $J$ represents the number of items in the item bank. 
    \vspace{-0.1cm}
    \item \textbf{Online Deployment}: Given the estimated item parameters, an item selection algorithm is deployed online to adaptively select items for future examinees.
\end{itemize}

 The performance of CAT can be measured by the number of items required to estimate an examinee’s latent traits with sufficient precision.

\subsection{Notation and Problem Formulation} \label{subsec:notations}

For the entirety of the paper, we assume the calibration dataset $\bm{Y}$ is binary, where each element $y_{ij}$ represents whether subject $i$ answered item $j$ correctly. We consider a general two-parameter MIRT framework with $K$ latent factors \citep{bock2021item}. Let $\bm{B} \in \mathbbm R^{J \times K}$ denote the factor loading matrix, and $\bm{D} \in \mathbbm R^J$ denote the intercept vector. {\textcolor{black}{Throughout the paper, boldface notation is used to denote vectors and matrices (e.g., $\bm{\theta}, \bm{B}, \bm{D}$), while scalar quantity such as $D_{j}$ is written in standard font}}. For each examinee $i$ with multivariate latent trait $\bm{\theta}^{(i)} \in \mathbbm R^K$, the data-generating process for $y_{ij}$ is given by:
\begin{equation} \label{eq:probit-mirt}
    y_{ij} \mid \bm{\theta}^{(i)}, \bm{B}_j, D_j \sim \text{Bernoulli}(\Phi(\bm{B}_j'\bm{\theta}^{(i)} + D_j)),
\end{equation}
where $\bm{B}_j$ is the $j$-th row of $\bm{B}$, $D_j$ is the $j$-th entry of $\bm{D}$, and $\Phi(\cdot)$ denotes the standard normal cumulative distribution function. The item parameters can be compactly expressed as $\{\bm{\xi}_j\}_{j=1}^J := \{(\bm{B}_j, D_j)\}_{j=1}^J$. This two-parameter MIRT framework is highly general, imposing no structural constraints on the loading matrix and requiring no specific estimation algorithms for item parameter calibration.

Given the estimated item parameters, we need to design an item selection algorithm that efficiently tests a future examinee with unobserved latent ability $\bm{\theta} \in \mathbbm R^K$. The sequential nature of the CAT problem makes Bayesian approaches particularly appealing. Without loss of generality, we assume a standard multivariate Gaussian prior on $\bm{\theta} \sim \mathcal{N}(0, \mathbbm I_K)$, and introduce the following notation:

\begin{itemize}
    \item For any positive integer $n$, let $[n]$ denote all the positive integers no greater than $n$. Let $j\in [J]$ be the index for the items in the item bank.

    \item Let $j_t$ denote the index of the item selected at step $t$ based on an arbitrary item selection rule. Define $\mathcal{I}_{t}$ as the set of the first $t$ administered items, and let $R_{t} := [J] \setminus \mathcal{I}_{t-1} $, the set of available items before the $t$-th item is picked.

    \item We use the shorthand notation $f(\bm{\theta} | \bm{Y}_{1:T})$ to represent the latent factor posterior distributions $f(\bm{\theta}| \bm{Y}_{1:T}, \bm{\xi}_{1:T})$ after $T$ items have been selected. Here, the response history is denoted by $\bm{Y}_{1:T}:= [y_{j_1}, \cdots, y_{j_T}]'$, and the item parameters are given by $\bm{\xi}_{1:T} := (\bm{B}_{1:T}, \bm{D}_{1:T})$, where $\bm{B}_{1:T}:=[\bm{B}_{j_1}, \cdots, \bm{B}_{j_T}]'$ and $\bm{D}_{1:T}:=[D_{j_1}, \cdots, D_{j_T}]'$. 
\end{itemize}

At time $(t+1)$, the algorithm takes the current posterior $f(\bm{\theta}|\bm{Y}_{1:t})$ as input and outputs the next item selection $j_{t+1}$. This process continues until at time $T'$, either when the posterior variance of $f(\bm{\theta}|\bm{Y}_{1:T'})$ falls below a predefined threshold $\tau^2$, or when $T'=H$, where $H \leq J$ is the maximum number of items that can be administered. Hence the goal of CAT is to minimize $T'$ given the estimated item parameters of the item bank.

\subsection{Reviews of KL Information Item Selection Rules} \label{subsec:existing_rules}

Given that our proposed deep CAT system is built on a general Bayesian MIRT framework \citep{doi:10.1080/01621459.2025.2476786}, we briefly revisit common Bayesian item selection rules and show in Section \ref{subsec:accelerate_benchmarks} how our framework can be used to accelerate these baseline methods. The popular KL Expected A Priori (EAP) rule selects the $t$-th item based on the average KL information between response distributions on the candidate item at the EAP estimate $\hat{\bm{\theta}}_{t-1} = \int_{\bm{\theta}} \bm{\theta} f(\bm{\theta} | \bm{Y}_{1:(t-1)}) d\bm{\theta}$,  and random factor $\bm{\theta}$ sampled from the posterior distribution $f(\bm{\theta}| \bm{Y}_{1:(t-1)})$ {\textcolor{black} {\citep{pekl}}}: \vspace{-0.2cm}  
\begin{equation} \label{eq:kl}
    \arg\max_{j_t \in R_t} \int_{\bm{\theta}} \bigg\{ \sum_{l=0}^{1} P(y_{j_t}=l | \hat{\bm{\theta}}_{t-1}) \log \frac{ P(y_{j_t}=l | \hat{\bm{\theta}}_{t-1})}{ P(y_{j_t}=l | \bm{\theta})} \bigg\} f(\bm{\theta} | \bm{Y}_{1:(t-1)}) d\bm{\theta}. 
\end{equation}

Rather than focusing on the KL information based on the response distributions, {\textcolor{black}{\cite{kl-is}}} propose the MAX Pos approach by maximizing the KL information between two subsequent latent factor posteriors $f(\bm{\theta}| \bm{Y}_{1:(t-1)})$ and $f(\bm{\theta} | \bm{Y}_{1:t})$. Intuitively, this approach prioritizes items that induce the largest shift in the posterior, formalized as: 
\begin{equation}\label{eq:sp_kl}
    \arg\max_{j_t \in R_t} \sum_{y_{j_t}=0}^{1} f(y_{j_t}| \bm{Y}_{1:(t-1)}) \int_{\bm{\theta}} f(\bm{\theta}|\bm{Y}_{1:(t-1)}) \log \frac{f(\bm{\theta}|\bm{Y}_{1:(t-1)})}{f(\bm{\theta}| \bm{Y}_{1:t})} d\bm{\theta},
\end{equation}
where $f(y_{j_t}| \bm{Y}_{1:(t-1)})$ represents the posterior predictive probability:
\begin{equation} \label{eq:pos_pred}
    f(y_{j_t}| \bm{Y}_{1:(t-1)}) = \int_{\bm{\theta}} f(y_{j_t}| \bm{\theta})  f(\bm{\theta}| \bm{Y}_{1:(t-1)}) d\bm{\theta}.
\end{equation}
A third Bayesian strategy, the mutual information (MI) approach, maximizes the mutual information between the current posterior distribution and the response distribution of new item $y_{j_{t}}$ \citep{mi_selection}. Rooted in experimental design theory \citep{Rnyi1961OnMO}, we can also interpret mutual information as the entropy reduction of the current posterior distribution of $\bm{\theta}$ after observing new response $y_{j_t}$. More formally:
\begin{equation} \label{eq:mi_eq}
    \arg\max_{j_t \in R_t} I_M (\bm{\theta}, y_{j_t}) = \arg\max_{j_t \in R_t} \sum_{y_{j_t}=0}^{1} \int_{\bm{\theta}} f(\bm{\theta}, y_{j_t} | \bm{Y}_{1:(t-1)}) \log \frac{f(\bm{\theta}, y_{j_t} | \bm{Y}_{1:(t-1)})}{f(\bm{\theta}| \bm{Y}_{1:(t-1)}) f(y_{j_t} | \bm{Y}_{1:(t-1)})} d\bm{\theta}. 
\end{equation}

Given the impressive empirical success of the MI method \citep{selection_survey}, we also introduce a competitive heuristic item-selection rule that chooses items with the highest predictive variance under the current posterior estimates, serving as an additional benchmark method. Formally, write $c_{j_t} := f(y_{j_t} | \bm{Y}_{1:(t-1)})$ as defined in (\ref{eq:pos_pred}), and consider the following criterion:
\begin{equation} \label{eq:var-pred-mean}
\arg\max_{j_t \in R_t} \int_{\bm{\theta}} (\Phi(\bm{B}_{j_t}' \bm{\theta} + D_{j_t})-c_{j_t})^2 f(\bm{\theta}|\bm{Y}_{1:(t-1)}) d\bm{\theta}. 
\end{equation}
This rule selects the item $j_t$ that maximizes the variance of the predictive means, weighted by the current posterior $f(\bm{\theta}|\bm{Y}_{1:(t-1)})$. In  Appendix \ref{sec:connection}, we establish a connection between this selection rule and the MI method, showing that both favor items with high prediction uncertainty.

\subsection{Reinforcement Learning Formulation} \label{subsec:rl_problem}

Existing item selection rules face several limitations: first, these methods rely on one-step-lookahead selection, choosing items that provide the most immediate information without considering their impact on future selections, which can result in suboptimal policies \citep{Sutton1998}. Second, they are heuristically designed to balance information across all latent factors, rather than emphasizing the most essential factors of interest. Finally, these rules do not directly minimize test length, potentially increasing test duration without proportional gains in accuracy.

A more principled approach is to formulate item selection as a reinforcement learning (RL) problem, where the optimal policy is learned using Bellman optimality principles rather than relying on heuristics that ignore long-term planning. Beyond addressing myopia, an RL-based formulation provides a direct mechanism to minimize the number of items required to reach a predefined posterior variance reduction threshold, explicitly guiding item selection toward accurately measuring the primary factors of interests.

More generally, consider a general finite horizon setting where each examinee can answer at most $H \leq J$ items, and the CAT algorithm terminates whenever the posterior variance of all factors of interest is smaller than a certain threshold $\tau^2$, or when $H$ is reached. Since $H$ can be set sufficiently large, it serves as a practical secondary stopping criterion to prevent excessively long tests. Formally:

\begin{itemize}
    \item \textbf{State Space} $\mathcal{S}$: the space of all possible latent factor posteriors $\mathcal{S} := \{f(\bm{\theta} | \bm{Y}_{1:t}): t \in \{1, \cdots, H\}\}.$
   The state variable in CAT represents the Bayesian estimate of the examinee's latent traits as a multivariate distribution at time $t$. We may write $s_t := f(\bm{\theta} | \bm{Y}_{1:t})$, with $s_0 := \mathcal{N}(0, \mathbbm{I}_K)$. The state space is discrete and can be exponentially large. Since there are $J$ items in the item bank and the responses are binary, the total number of possible states is $\sum_{t=1}^H \binom{J}{t} \cdot 2^t$. 

     \item \textbf{Action space} $\mathcal{A}$: the item bank $[J]:= \{1, \cdots, J\}$. However, at the t-th selection time, the action space is the remaining items in the test bank that have not yet been selected $R_t:= [J] \setminus \mathcal{I}_{t-1}$, as we do not select the same item twice. 

     \item \textbf{Transition Kernel} $\mathcal{P}: \mathcal{S} \times \mathcal{A} \rightarrow \Delta(\mathcal{S})$, where $\Delta(\mathcal{S})$ denotes the space of probability distributions over $\mathcal{S}$. The transition kernel $\mathcal{P}$ specifies the stochastic rule that governs how the state evolves after an action is taken. In the CAT context, the kernel maps the current latent-trait estimate $s_t$ and the selected item $a_t$ to the next posterior state $s_{t+1}$. Conditional on $(s_t, a_t)$, the next estimate $s_{t+1}$ depends on the examinee’s binary response $y_t \in \{0,1\}$ to item $a_t$, which yields only two possible posterior updates. Because the true ability of the examinee and thus the response probability $p(y_t = 1 \mid s_t, a_t)$ are unknown, the transition kernel in CAT is also unknown and must be approximated.

    \item \textbf{Reward Function}: To directly minimize the test length, we assign a simple $0-1$ reward structure, where we assign negative rewards whenever more items are needed to reduce posterior variance to a given threshold. In a $K$-factor MIRT model, we often prioritize a subset of factors $\mathcal{K} \subset [K]$ (e.g. $\mathcal{K} = \{1\}$ for pCAT-COG), with the test terminating once the posterior variances of all factors in $\mathcal{K}$ fall below the predefined threshold $\tau^2$:
        \begin{equation}
        \label{eq:reward}
        R^{(t)}(s_t, a, s_{t+1}) = 
        \begin{cases} 
        -1 & \text{if } V_{t+1} > \tau^2, \\
        0 & \text{otherwise},
        \end{cases}
        \end{equation}
        where \(V_{t+1} = \max_{k \in \mathcal{K}} \operatorname{Var}(\bm{\theta}_k \mid \bm{Y}_{1:(t+1)}) \) is the maximum marginal posterior variance among the prioritized factors $\mathcal{K}$. This also simplifies learning the value function, as the rewards are always bounded integers within $[-H, -1]$. We further illustrate the advantages of adopting this reward structure in Appendix \ref{sec:greedy-proof}.
\end{itemize}

Rather than following a heuristic rule, we learn a policy $\pi$ that maps the current state (latent factor estimates) to a distribution over potential items: 
\[
\pi:\ \mathcal{S}\to\Delta(\mathcal{A}),\qquad a_t \sim \pi(\cdot \mid s_t).
\]
Given an initial state $S_0=s_0$, and the discount factor $\gamma \in (0,1]$, the value function is
\begin{equation} \label{eq:value-function}
    v_{\pi}(s_0) := \mathbbm{E}_{\pi} \left[\sum_{t=0}^{H-1} \gamma^t R^{(t)}(s_t, \pi(s_t), s_{t+1}) | S_0 = s_0 \right],
\end{equation}
The expectation $\mathbb{E}_\pi[\cdot]$ is taken over trajectories generated by drawing $a_t$ from $\pi(\cdot \mid s_t)$ and then $s_{t+1}$ from the transition kernel $\mathcal{P}(\cdot \mid s_t,a_t)$. 
Because direct evaluation of (\ref{eq:value-function}) is challenging, it is more convenient to consider the action–value function
\begin{equation} \label{eq:q-function}
    Q_{\pi}(s_0, a) := \mathbbm E_{\pi} \left[\sum_{t=0}^{H-1} \gamma^{t}R^{(t)} (s_t, \pi(s_t), s_{t+1}) | S_0=s_0, A_0= a \right],
\end{equation}
with $v_\pi(s)=\mathbb{E}_{a\sim\pi(\cdot\mid s)}Q_\pi(s,a)$. The optimal policy $\pi^{\star}$ satisfies the Bellman equation \citep{10.5555/560669}:
\begin{equation} \label{eq:bellman}
    Q_{\pi^\star}(s,a)\;=\;\mathbb{E}_{S'\sim\mathcal{P}(\cdot\mid s,a)}\!\left[\,R(s,a,s')+\gamma\,\max_{a'} Q_{\pi^\star}(s',a')\,\right],
\end{equation}
where $s'$ is the next posterior estimate after applying action $a$ to the current posterior $s$, and $v_{\pi^\star}(s)=\max_{a} Q_{\pi^\star}(s,a)$. Since the state space grows exponentially and the transition kernel is unknown, solving for $\pi^*$ using traditional dynamic programming approaches becomes intractable \citep{Sutton1998}. 

In practice, our deep Q-learning approach \citep{Mnih2013PlayingAW} does not evaluate these expectations analytically. Let the transition $(s,a,r,s')$ represent one testing step, where item $a$ is selected under state $s$, the reward $r$ is observed, and the posterior is updated to $s'$ through the Bayesian MIRT model. We approximate the Bellman fixed point in (\ref{eq:bellman}) by fitting a parametric function $Q_w$ that minimizes the squared temporal-difference loss:
\begin{equation} \label{eq:td-loss}
    \mathcal{L}(w)
= \mathbb{E}_{(s,a,r,s')\sim \mathcal{D}}
\big[\big(Q_w(s,a) - y(s,a,r,s')\big)^2\big]
\end{equation}
where
\begin{equation} \label{eq:pred-q}
y(s,a,r,s') = r + \gamma \max_{a'} Q_{\bar{w}}(s',a').
\end{equation}
Here, $\mathcal{D}$ (replay buffer) denotes a collection of previously observed testing steps $(s,a,r,s')$ generated during simulation, which is used to approximate the expectation above by empirical averaging. The parameter $\bar{w}$ corresponds to a delayed copy of the model parameters $w$, updated less frequently to stabilize the numerical optimization. We describe the full Q-learning algorithm in Section~\ref{sec:Q-rl}.

\section{Accelerating Item Selection via Posterior Identification} \label{sec:acceleration}

The central insight of our deep CAT framework is that, by iteratively applying the E-step of the PXL-EM algorithm \citep{doi:10.1080/01621459.2025.2476786}, the latent factor posteriors admit tractable posterior updates. This result is critical for reinforcement learning, as the posterior distribution is deemed to be non-Gaussian and is analytically intractable under the traditional MIRT literature. By obtaining a tractable representation of this posterior, we can parametrize the examinee’s evolving ability and uncertainty as a well-defined state variable, which is an essential prerequisite for applying Q-learning to adaptive testing. The Bayesian MIRT formulation thus not only replaces costly MCMC procedures with efficient posterior updates under a probit link for existing Bayesian item selection rules discussed in Section \ref{subsec:existing_rules}, but also supplies the statistical foundation that makes the subsequent reinforcement-learning framework feasible.

Specifically, we show that the latent factor posterior updates during CAT belong to an instance of the unified-skew-normal distribution \citep{unified_skew_normal}, defined as follows:
\begin{definition} \label{def:usn}
Let $\Phi_{T} \left\{ \bm{V}; \bm{\Sigma} \right\}$ represent the cumulative distribution function of a T-dimensional multivariate Gaussian distribution $N_{T}\left(0_{T}, \bm{\Sigma} \right)$ evaluated at vector $\bm{V}$. A K-dimensional random vector $\bm{\theta} \sim \operatorname{SUN}_{K, T}(\bm{\mu}, \bm{\Omega}, \bm{\Delta}, \bm{\gamma}, \bm{\Gamma})$ has the \textbf{unified skew-normal} distribution if it has the probability density function:
$$
\phi_{K}(\bm{\theta}; \bm{\mu}, \bm{\Omega}) \frac{\Phi_{T}\left\{\bm{\gamma}+\bm{\Delta}' \bar{\bm{\Omega}}^{-1} \bm{\omega}^{-1}(\bm{\theta}-\bm{\mu}) ; \bm{\Gamma}-\bm{\Delta}' \bar{\bm{\Omega}}^{-1} \bm{\Delta}\right\}}{\Phi_{T}(\bm{\gamma} ; \bm{\Gamma})}.
$$
Here,  $\phi_{K}(\bm{\theta} ; \bm{\mu}, \bm{\Omega})$ is the density of a $K$-dimensional multivariate Gaussian with expectation $\bm{\mu}=\left(\mu_{1}, \ldots, \mu_{K}\right)' $, and a $K$ by $K$ covariance matrix $\bm{\Omega} =\bm{\omega} \bar{\bm{\Omega}} \bm{\omega}$, where $\bar{\bm{\Omega}}$ is the correlation matrix and $\bm{\omega}$ is a diagonal matrix with the square roots of the diagonal elements of $\bm{\Omega}$ in its diagonal.  $\bm{\Delta}$ is a $K$ by $T$ matrix that determines the skewness of the distribution, and $\bm{\gamma} \in \mathbbm R^T$ controls the flexibility in departures from normality.

In addition, the $(K+T) \times(K+T)$ matrix $\bm{\Omega}^{*}$, having blocks $\bm{\Omega}_{[11]}^{*}=\bm{\Gamma}, \bm{\Omega}_{[22]}^{*}=\bar{\bm{\Omega}}$ and $\bm{\Omega}_{[21]}^{*}=\bm{\Omega}_{[12]}^{*'}=\bm{\Delta}$, needs to be a full-rank correlation matrix.
\end{definition}

Suppose an arbitrary CAT item selection algorithm has already selected $T$ items with item parameters $\bm{\xi}_{1:T} = (\bm{B}_{1:T}, \bm{D}_{1:T})$, where $\bm{B}_{1:T}:=[\bm{B}_{j_1}, \cdots, \bm{B}_{j_T}]'$ and $\bm{D}_{1:T}:=[D_{j_1}, \cdots, D_{j_T}]'$. Then it is possible to show the following result:

\begin{theorem}\label{thm:3.2}
Consider a K-factor CAT item selection procedure after selecting $T$ items, with $\mathcal{N}(\bm{0}_K, \mathbbm{I}_K)$ prior placed on the test taker's latent trait $\bm{\theta}$. If $\bm{Y}_{1:T}=\left(y_{j_1}, \ldots, y_{j_T}\right)'$ is conditionally independent binary response data from the two-parameter probit MIRT model defined in (\ref{eq:probit-mirt}), then
\begin{equation*}
   (\bm{\theta} \mid \bm{Y}_{1:T}, \bm{B}_{1:T}, \bm{D}_{1:T}) \sim \operatorname{SUN}_{K, T}\left(\bm{\mu}_{\mathrm{post}}, \bm{\Omega}_{\text {post }}, \bm{\Delta}_{\text {post }}, \bm{\gamma}_{\text {post }}, \bm{\Gamma}_{\text {post }}\right), 
\end{equation*}
with posterior parameters
\begin{gather*}
  	\bm{\mu}_{\mathrm{post}} = \bm{0}_K, \quad \bm{\Omega}_{\text {post }}= \mathbbm{I}_K,  \quad \bm{\Delta}_{\text {post }}=  \bm{C}_1'\bm{C}_3^{-1}, \\
	\bm{\gamma}_{\text {post }} = \bm{C}_3^{-1} \bm{C}_2, \quad \bm{\Gamma}_{\text {post }} =\bm{C}_3^{-1}\left(\bm{C}_1 \bm{C}_1'+\mathbbm I_{T}\right) \bm{C}_3^{-1},
   \end{gather*}
where $\bm{C}_1 = \text{diag}(2y_{j_1}-1, \cdots, 2y_{j_T}-1) \bm{B}_{1:T}$ and  $\bm{C}_2 = \text{diag}(2y_{j_1}-1, \cdots, 2y_{j_T}-1)\bm{D}_{1:T}$. The matrix $\bm{C}_3$ is a $T$ by $T$ diagonal matrix, where the $(t,t)$-th entry is $(\|\bm{B}_{t,T}\|_2^2 +1)^{\frac{1}{2}}$, with $\bm{B}_{t, T}$ representing the $t$-th row of $\bm{B}_{1:T}$. 
\end{theorem}

Theorem \ref{thm:3.2} provides an exact finite-sample Bayesian characterization of the latent factor posterior. With a multivariate normal prior on $\bm{\theta}$ and a probit MIRT likelihood, the posterior $f(\bm{\theta}\mid \bm{Y}_{1:T})$ belongs to the unified skew-normal (SUN) family. This does not conflict with the empirical Bayes result of \cite{ChangStout1993}, which establishes asymptotic posterior normality as the number of items $J$ grows; \textcolor{black}{related discussion of Bayesian latent-trait estimation and asymptotic covariance properties in MIRT is given in \cite{Wang_2015}.} The finite-sample form is particularly relevant for CAT, where only a small number of items has been administered and large-sample normal approximations may be unreliable. This representation enables exact posterior calculations, improving uncertainty quantification and item selection in short tests. 
\textcolor{black}{For ordinal responses under a probit link, closely related SUN-based approximations can also be developed as discussed in \cite{doi:10.1080/01621459.2025.2476786}. The latent factor posterior can often be well approximated within the SUN family, although the representation is generally not exact.}

According to \cite{unified_skew_normal}, an arbitrary unified skew normal distribution $\bm{\theta} \sim \operatorname{SUN}_{K, T}(\bm{\mu}, \bm{\Omega}, \bm{\Delta}, \bm{\gamma}, \bm{\Gamma})$ has a stochastic representation as a linear combination of a $K$-dimensional multivariate normal random variable $\bm{V}_0$, and a $T$-dimensional truncated multivariate normal random variable $\bm{V}_{1, -\gamma}$ as follows:
\begin{equation} \label{eq:sto_rep}
   \bm{\theta} \overset{d}{=} \bm{\mu} + \bm{\omega} (\bm{V}_0 + \bm{\Delta}  \bm{\Gamma}^{-1} \bm{V}_{1, - \gamma}), 
\end{equation}
where $\bm{V}_0 \sim \mathcal{N}(0, \bar{\bm{\Omega}} - \bm{\Delta} \bm{\Gamma}^{-1} \bm{\Delta}') \in \mathbbm R^{K}$, and $\bm{V}_{1, -\gamma}$ is a multivariate normal distribution $ \mathcal{N}(0, \bm{\Gamma})$ truncated to $\{\bm{V}_1 \in \mathbb{R}^{T}: V_{1i} \geq -\gamma_i \quad,\forall i\}$.

Based on Theorem \ref{thm:3.2}, we have a closed-form expression for the posterior parameters $(\bm{\mu}, \bm{\Omega}, \bm{\Delta}, \bm{\gamma}, \bm{\Gamma})$ for $\bm{\theta}$. Recall that $\bm{\Omega} = \bm{\omega} \bar{\bm{\Omega}} \bm{\omega}$ is the standard covariance correlation matrix decomposition from definition \ref{def:usn}, and hence the only unrealized stochastic terms in equation (\ref{eq:sto_rep}) are $\bm{V}_0$ and $\bm{V}_{1, -\bm{\gamma}}$. This suggests that sampling from the latent factor posterior distribution $f(\bm{\theta}|\bm{Y}_{1:T})$ requires two independent steps: sampling from a $K$-dimensional multivariate normal distribution $\bm{V}_0$, and a $T$-dimensional multivariate truncated normal distribution $\bm{V}_1$. As a result, the direct sampling approach scales efficiently with the number of factors $K$, since generating samples from $K$-dimensional multivariate normal distributions is trivial. Moreover, in CAT settings, the test is typically terminated early, meaning $T$ remains relatively small. When $T$ is moderate (e.g., $T<1,000$), sampling from the truncated multivariate normal distribution remains computationally efficient using the minimax tilting method \citep{Botev_2016}. The exact proof of Theorem \ref{thm:3.2} and the sampling details can be found in Appendix \ref{sec:proof}.

The direct sampling approach provides substantial gains in both computational efficiency and numerical precision compared with traditional MCMC methods commonly used in the MIRT literature \citep{BeguinGlas2001, JiangTemplin2019}. In standard MIRT settings, the posterior distribution of the latent factors is typically regarded as intractable and non-Gaussian, requiring a Markov chain to be constructed via data-augmentation techniques \citep{f3928584-48a1-3ab1-b404-2ec9ffdf51bb, Polson01122013} so that its stationary distribution approximates the posterior. This procedure entails repeated simulation of augmented data and is inherently sequential, which limits opportunities for parallelization. Moreover, it demands additional tuning, burn-in, and convergence diagnostics to ensure that the chain adequately converges to the posterior distribution. In contrast, Theorem \ref{thm:3.2} establishes that the posterior distribution can be expressed exactly as a unified skew-normal distribution, allowing direct and parallel sampling without the need for iterative convergence procedures.

Theorem \ref{thm:3.2} also plays a central role in our proposed deep Q-learning algorithm, as it fully characterizes the latent factor posterior distribution $f(\bm{\theta}|\bm{Y}_{1:T})$, enabling the parametrization of the state variable $s_T$, which serves as an input to the Q-network illustrated in Section \ref{subsec:network-design}. We also illustrate the application of Theorem \ref{thm:3.2} in accelerating the existing item selection rules below.

\subsection{Accelerating Existing Rules} \label{subsec:accelerate_benchmarks}

While Bayesian CAT criteria require multidimensional integration, these quantities are typically evaluated in practice via Monte Carlo approximation rather than analytic quadrature in moderate to high dimensions. In such Bayesian CAT implementations, the dominant computational cost lies in obtaining valid samples from the latent factor posterior. By providing exact posterior samples without Markov chain construction, our approach removes this primary bottleneck and renders the remaining Monte Carlo integration computationally efficient and easily parallelizable. For example, in computing the KL-EAP item selection rule (\ref{eq:kl}), we directly sample from $f(\bm{\theta}|\bm{Y}_{1:(t-1)})$ instead of resorting to MCMC, and evaluate the integral via Monte Carlo integration. Since the EAP estimate $\hat{\bm{\theta}}_{t-1}$ remains fixed at time step $t$, evaluating the KL information term is straightforward. For the Max Pos item selection rule in equation (\ref{eq:sp_kl}), we again obtain i.i.d. samples from $f(\bm{\theta}|\bm{Y}_{1:(t-1)})$ and compute the posterior predictive probabilities in (\ref{eq:pos_pred}). Although the density ratio $\frac{f(\bm{\theta}|\bm{Y}_{1:(t-1)})}{f(\bm{\theta}|\bm{Y}_{1:t})}$ is difficult to evaluate, we can leverage the conditional independence assumption of the MIRT model. In particular, we can express the joint distribution of $(\bm{\theta}, y_{j_t})$ as 
\begin{equation} \label{eq:mirt-trick}
    f(\bm{\theta}, y_{j_t} | \bm{Y}_{1:(t-1)}) =  f(\bm{\theta} | \bm{Y}_{1:(t-1)}, y_{j_t}) f(y_{j_t} | \bm{Y}_{1:(t-1)}) = f(y_{j_t}| \bm{\theta}) f(\bm{\theta} | \bm{Y}_{1:(t-1)}).
\end{equation}
Using Equation (\ref{eq:mirt-trick}), we rewrite the KL information term in Equation (\ref{eq:sp_kl}) as:
\begin{align*}
  \int f(\bm{\theta}| \bm{Y}_{1:(t-1)}) \log \frac{f(\bm{\theta}| \bm{Y}_{1:(t-1)})}{f(\bm{\theta}| \bm{Y}_{1:(t-1)}, y_{j_t})} d\bm{\theta}  &  = \int f(\bm{\theta}| \bm{Y}_{1:(t-1)}) \log  \frac{f(\bm{\theta}| \bm{Y}_{1:(t-1)}) f(y_{j_t}| \bm{Y}_{1:(t-1)})}{f(\bm{\theta}, y_{j_t} | \bm{Y}_{1:(t-1)})} d\bm{\theta}  \\
  & = \int f(\bm{\theta}| \bm{Y}_{1:(t-1)}) \log  \frac{f(y_{j_t}| \bm{Y}_{1:(t-1)})}{f(y_{j_t} | \bm{\theta})} d\bm{\theta}.
\end{align*}
Since $f(y_{j_t}|\bm{\theta})$ can be easily computed for each $\bm{\theta} \sim f(\bm{\theta}| \bm{Y}_{1:(t-1)})$, the online computation of Max Pos remains efficient. 

Although the mutual information selection rule in Equation (\ref{eq:mi_eq}) has demonstrated strong empirical performance \citep{selection_survey}, its computational complexity remains a significant challenge. By applying Equation (\ref{eq:mirt-trick}), we can rewrite mutual information as
\begin{equation} \label{eq:mi-rewrite}
     \arg\max_{j_t \in R_t} \sum_{y_{j_t}=0}^{1}  f(y_{j_t}| \bm{Y}_{1:(t-1)}) \int_{\bm{\theta}} f(\bm{\theta}| \bm{Y}_{1:(t-1)}, y_{j_t} ) \log \frac{f(y_{j_t} | \bm{\theta})}{ f(y_{j_t} | \bm{Y}_{1:(t-1)})} d\bm{\theta}.
\end{equation}
This formulation reveals that maximizing mutual information is structurally similar to Max Pos but requires much more computational effort. Unlike Max Pos, where sampling is performed from the current posterior $f(\bm{\theta}|\bm{Y}_{1:(t-1)})$, mutual information requires sampling from future posteriors $f(\bm{\theta}| \bm{Y}_{1:(t-1)}, y_{j_t})$. Since each candidate item $j_t \in R_t$ has two possible outcomes ($y_{j_t} = 0$ or $y_{j_t} = 1$), evaluating equation (\ref{eq:mi-rewrite}) requires obtaining samples from $|R_t| \times 2$ distinct posterior distributions. Even if sampling each individual posterior is computationally efficient, this approach becomes impractical for large item banks.

We hence propose a new approach to dramatically accelerate the computation of the mutual information quantity using the idea of importance sampling and resampling and bootstrap filter \citep{Gordon1993NovelAT,sir}. Rather than explicitly sampling the future posterior $f(\bm{\theta}| \bm{Y}_{1:(t-1)}, y_{j_t})$ for each $j_t \in R_t$ and $y_{j_t} \in \{0, 1\}$, we can simply sample from the current posterior $f(\bm{\theta} | \bm{Y}_{1:(t-1)})$ once, and then perform proper posterior reweighting. Under equation (\ref{eq:probit-mirt}), let $p(\bm{\theta})$ be the prior on the latent factors $\bm{\theta}$,  $l_1(\bm{\theta})$ be the current data likelihood, and $l_2(\bm{\theta})$ denote the future data likelihood after observing $y_{j_t}$. We have the future posterior density as follows:
\begin{align} \label{eq:reweight}
     f(\bm{\theta}| \bm{Y}_{1:(t-1)}, y_{j_t}) & \propto l_2(\bm{\theta}) p(\bm{\theta}) \propto \frac{l_2(\bm{\theta})p(\bm{\theta})}{l_1(\bm{\theta})p(\bm{\theta})} f(\bm{\theta} | \bm{Y}_{1:(t-1)}) \notag \\
     & = \Phi(\bm{B}_{j_t}'\bm{\theta} + D_{j_t})^{y_{j_t}}(1- \Phi(\bm{B}_{j_t}'\bm{\theta} + D_{j_t}))^{1- y_{j_t}} f(\bm{\theta} | \bm{Y}_{1:(t-1)}).
\end{align}
Equation (\ref{eq:reweight}) suggests we can generate samples from $f(\bm{\theta}| \bm{Y}_{1:(t-1)}, y_{j_t})$ via reweighting and resampling from the current posterior samples $f(\bm{\theta}|\bm{Y}_{1:(t-1)})$. Specifically, given a sufficiently large set of posterior samples $(\bm{\theta}_1, \cdots, \bm{\theta}_M) \sim f(\bm{\theta}|\bm{Y}_{1:(t-1)})$, we assign distinct weights for each sample $m \in \{1, \cdots, M\}$ and future item $j_t \in R_t$: 
$$q_m = \frac{w_m}{\sum_{i=1}^M w_i}, \quad \text{where } w_i= \Phi(\bm{B}_{j_t}'\bm{\theta}_i + D_{j_t})^{y_{j_t}}(1- \Phi(\bm{B}_{j_t}'\bm{\theta}_i + D_{j_t}))^{1- y_{j_t}}.$$
To sample from $f(\bm{\theta}|\bm{Y}_{1:(t-1)}, y_{j_t})$, we then draw from the discrete distribution over $(\bm{\theta}_1, \cdots, \bm{\theta}_M)$, placing weight $q_m$ on $\bm{\theta}_m$. This approach eliminates the need to sample from $|R_t| \times 2$ distinct posteriors directly, further accelerating mutual information computation and making it scalable for large item banks.

\section{Learning Optimal Item Selection Policy} \label{sec:Q-rl}

Building on the Bayesian MIRT foundation established in the previous section, we now formulate the problem of CAT as an RL task. The Bayesian framework provides two key advantages that make this integration both principled and computationally feasible. First, the identified latent factor posterior distributions offer a well-defined representation of examinee knowledge, allowing their corresponding posterior parameters to be parametrized directly as state variables. Without such identification, encoding an unknown and analytically intractable posterior would be ambiguous. Second, because reinforcement learning typically requires extensive simulations to learn an optimal policy, the acceleration achieved in online item selection enables rapid simulation of testing sessions, thereby substantially improving the efficiency of policy training.

As illustrated in Section \ref{subsec:rl_problem}, an RL approach addresses the myopic nature of traditional CAT selection rules, enables a more flexible reward structure, and directly minimizes the number of items required for performing online adaptive testing. Specifically, we propose a novel double Q-learning algorithm \citep{Mnih2015, 
van_Hasselt_Guez_Silver_2016} for online item selection in CAT. The algorithm trains a deep neural network offline using only the item parameters estimated from an arbitrary two-parameter MIRT model. {\textcolor{black}{The offline training phase is completed before any live CAT administration. After training, the network weights are fixed, and online CAT only requires sequentially updating the posterior state and applying the learned network to select the next item.}} During online item selection, the neural network takes the current posterior distribution \( s_t := f(\bm{\theta} \mid \bm{Y}_{1:t}) \) 
as input and outputs the next item selection for step \( (t+1) \). The neural network architecture is described in Section \ref{subsec:network-design}, while the double deep Q-learning algorithm is detailed in Section \ref{subsec:doubleQ}.

\subsection{Deep Q-learning Neural Network Design} \label{subsec:network-design}

A potential reason that a deep reinforcement learning approach has not been proposed in the CAT literature is the ambiguity arising from unidentified latent factor distributions. However, by Theorem \ref{thm:3.2}, it is straightforward to compactly parametrize the posterior distribution at each time step $t$ using the parameters $\bm{C}_1 \in \mathbbm R^{t \times K}$, $\bm{C}_2 \in \mathbbm{R}^t$, and $\text{diag}(\bm{C}_3) \in \mathbbm R^t$, where $\text{diag}(\bm{C}_3)$ represents the diagonal vector of $\bm{C}_3$. This structured representation of the posterior enables item selection policy learning via deep neural networks, which takes posterior parameters as input and outputs item selection.

Neural networks can be regarded as flexible function approximators that learn nonlinear mappings between inputs and outputs through a composition of simple transformations \citep{Goodfellow-et-al-2016}. In our framework, the deep neural network consists of two key components: an encoder and a classifier. At each time step $t$, the encoder maps the collection of posterior parameters to a latent representation in $\mathbb{R}^{L}$, where $L$ is a hyperparameter and represents the dimension of the latent feature space. The classifier then outputs a $J$-dimensional vector of estimated Q-values and selects the next item from the item bank that is expected to yield the highest Q-value.

Define the collection of the posterior parameters $\tilde{\bm{\xi}}_t:= \{(\bm{C}_{1h}, C_{2h}, C_{3h})\}_{h=1}^{t}$ , where $\bm{C}_{1h} \in \mathbb{R}^{K}$ is the $h$-th row of $\bm{C}_1$, $C_{2h} \in \mathbb{R}$ is the $h$-th element of $\bm{C}_2$, and $C_{3h} \in \mathbb{R}$ is the h-th elements in the diagonal of $\bm{C}_3$. Since the size of the posterior parameters $\tilde{\bm{\xi}}_t$ grows over time, and permuting the tuples within $\tilde{\bm{\xi}}_t$ still describes the same posterior, we have to design a neural network that can take inputs of growing size, and can provide output that is permutation invariant of the inputs. One solution is to consider 
the weight sharing idea from the Bayesian experimental design literature \citep{Foster2021DeepAD}. Let $\phi_1(.): \mathbb{R}^{K+2} \rightarrow \mathbb{R}^{L_1}$ denote an encoder component that maps each tuple $(\bm{C}_{1h}, C_{2h}, C_{3h}) \in \mathbb{R}^{K+2}$ to an $L_1$-dimensional latent space, and consider the operation $g_1(.)$ as follows:
\vspace{-0.2cm}
\begin{equation} \label{eq:permu-invariant}
     g_1(\tilde{\bm{\xi}}_t) := \sum_{h=1}^t \phi_1 \{(\bm{C}_{1h}, C_{2h}, C_{3h})\}.
\end{equation}
Observe that $g_1(.)$ is capable of handling a growing number of inputs through summations of $\phi_1(.)$ functions over $t$. More importantly, permuting the order of the tuples in $\tilde{\bm{\xi}}_t$ does not change the value of $g_1(.)$, since summation is permutation invariant. Although the form of $g_1$ may look restrictive,  any function $f(.)$ operating on a countable set can be decomposed into the form $\rho \circ g_1(.)$, where $\rho(.)$ is a suitable transformation that can be learned from another neural network (see Theorem 2 of \cite{10.5555/3294996.3295098}). 

In principle, the network design $\rho \circ g_1(.)$, coupled with the Q-learning algorithm, is sufficient to learn the optimal policy. However, to enhance learning efficiency, we further enrich the state representation $\tilde{\bm{\xi}}_t$ with a matrix of prediction quartiles $\bm{\Psi}_t \in \mathbbm R^{J \times Q}$, where each row contains a vector of quantiles of the predictive distribution for item $j$. We form these quantiles by first drawing samples $\bm{\theta}_{1}, \cdots, \bm{\theta}_M \sim f(\bm{\theta} | \bm{Y}_{1:t})$ as in Section \ref{sec:acceleration}, and then computing the quantiles of the prediction samples $\{\Phi(\bm{B}_j'\bm{\theta}_i + D_j)\}_{i=1}^M$ for each item $j$. Since sampling from $f(\bm{\theta}| \bm{Y}_{1:t})$ only needs to be done once, computing the matrix $\bm{\Psi}_t$ online is computationally efficient.

This practice of augmenting the raw state variable with additional contextual features of the states ($\bm{\Psi}_t$) echoes the common strategy in the RL literature \citep{Mnih2015,pmlr-v37-schaul15}: learning tends to be more stable and efficient when the state representation captures not only the current ability estimate but also characteristics of the item bank. In our setting, $\bm{\Psi}_t$ offers a richer description of the predictive distributions across items given the current estimates, effectively serving as additional covariates for item selection. Empirically, we find that incorporating $\bm{\Psi}_t$ substantially accelerates convergence toward the optimal item selection policy. Including $\bm{\Psi}_t$ in the state representation is harmless, as it is a deterministic function of the posterior distribution; identical posteriors will always produce identical $\bm{\Psi}_t$, ensuring that the policy remains consistent across equivalent states.

In summary, Figure \ref{fig:policy-architecture} describes the architecture of the policy network used in our approach. \textcolor{black}{The same network is trained offline and then deployed during live CAT with fixed weights, so only the posterior state is updated sequentially during online administration}. To select the $(t+1)$-th item, our proposed neural net takes two inputs: the posterior parameters $\tilde{\bm{\xi}}_t$ and the prediction matrix $\bm{\Psi}_t$, and outputs the selected item. Let $\phi_2: \mathbb{R}^{J \times Q} \to \mathbb{R}^{L_2}$ represent the encoder component that maps the matrix $\bm{\Psi}_t \in \mathbbm R^{J \times Q}$ to $L_2$-dimensional space. Write $L=L_1+L_2$ and the concatenated outputs of $g_1(.)$ and $\phi_2(.)$ as $[g_1(\tilde{\bm{\xi}}_t)', \phi_2(\bm{\Psi}_t)']' \in \mathbb{R}^L $ , we then define the classification component of neural network as $\rho(.): \mathbb{R}^L \rightarrow \mathbb{R}^J$, which maps the concatenated outputs of $g_1(.)$ and $\phi_2(.)$ to a $J$-dimensional logit vector. Denote the final policy network as $\pi_{\phi}$ and recall that $\rho(.)$ represents the classification layer, our proposed network selects the item at time $t$ corresponding to the maximum value of the following function:
\begin{equation*}
     \pi_{\phi}(\tilde{\bm{\xi}}_t, \bm{\Psi}_t) := \rho \circ \left\{ [g_1(\tilde{\bm{\xi}}_t)', \phi_2(\bm{\Psi}_t)' ]' \right\}.
\end{equation*}
We implement both primary and target Q-networks as simple feed-forward multilayer perceptron and find that their performance is largely insensitive to the exact choice of $L_1$ and $L_2$, and does not require very large depth. All architectural details such as layer sizes, activations, and optimizer settings are provided in Appendix \ref{sec:network_details}. This robustness suggests that our empirical improvements stem from the CAT‐specific state design and augmentation rather than network complexity.

% Preamble:
% \usepackage{tikz}
% \usetikzlibrary{arrows.meta,positioning,calc}
% \tikzset{>={Stealth[length=6pt,inset=2pt]}}

\begin{figure}[t]
\centering
\resizebox{0.96\linewidth}{!}{%
\begin{tikzpicture}[
    font=\small,
    node distance=10mm and 14mm,
    every node/.style={align=center},
    inb/.style  ={draw=black!60, rounded corners, fill=yellow!12, inner sep=4pt},
    enc/.style  ={draw=black!55!black, rounded corners, fill=black!10, inner sep=5pt},
    agg/.style  ={draw=black!60, circle, fill=gray!10, minimum size=9mm, inner sep=0pt},
    cat/.style  ={draw=purple!60!black, rounded corners, fill=purple!8, inner sep=4pt},
    cls/.style  ={draw=green!50!black, rounded corners, fill=green!12, inner sep=5pt},
    thinarrow/.style={->, semithick}
]

% --- Top branch: tuples -> phi1 -> sum -> g1
\node[inb] (xi) {$\{(\bm{C}_{1h},\,C_{2h},\,C_{3h})\}_{h=1}^{t}$\\\footnotesize posterior tuples};
\node[enc, right=14mm of xi] (phi1) {$\phi_1(\cdot)$\\\footnotesize encoder\\\footnotesize(shared over $h$)};
\node[agg, right=18mm of phi1] (sum) {$\sum$};
\node[inb, right=14mm of sum] (g1) {$g_1(\tilde{\boldsymbol{\xi}}_t)\in\mathbb{R}^{L_1}$};

\draw[thinarrow] (xi) -- node[above,sloped]{\footnotesize apply per tuple} (phi1);
\draw[thinarrow] (phi1) -- (sum);
\draw[thinarrow] (sum) -- node[above]{\footnotesize perm. invariant} (g1);
\node[below=1mm of sum] {\footnotesize sum over $h$};

% --- Bottom branch: Psi -> phi2
\node[inb, below=20mm of phi1] (psi) {$\bm{\Psi}_t\in\mathbb{R}^{J\times Q}$\\\footnotesize prediction quantiles};
\node[enc, right=14mm of psi] (phi2) {$\phi_2(\cdot)$\\\footnotesize encoder};
\draw[thinarrow] (psi) -- (phi2);

% --- Concatenate placed in the vertical middle between g1 and phi2
\node[cat] (concat) at ($ (g1.east)!0.5!(phi2.east) + (18mm,0) $)
    {$[\,g_1;\ \phi_2\,]\in\mathbb{R}^{L_1+L_2}=\mathbb{R}^{L}$};

\draw[thinarrow] (g1) -- (concat);
\draw[thinarrow] (phi2) -- (concat);

% --- Classifier rho and logits, centered to the right
\node[cls, right=16mm of concat] (rho) {$\rho(\cdot)$\\\footnotesize classifier};
\node[inb, right=14mm of rho] (logits) {item $j$};

\draw[thinarrow] (concat) -- (rho);
\draw[thinarrow] (rho) -- (logits);

\end{tikzpicture}%
}
\caption{High-level architecture of the Q-network. The shared encoder $\phi_1$ maps each tuple of posterior parameters to $\mathbb{R}^{L_1}$ and the sum yields the permutation invariant representation $g_1(\tilde{\boldsymbol{\xi}}_t)$. The matrix $\bm{\Psi}_t$ is encoded by $\phi_2$. The concatenated vector in $\mathbb{R}^{L}$ is passed to the classifier $\rho$ to select the $j$-th item (largest value in the $J$ logits). This network is trained offline using Algorithm \ref{alg:q-learning}; during live CAT, its weights are fixed and only the posterior state is updated sequentially.}
\label{fig:policy-architecture}
\end{figure}

\subsection{Double Deep Q-learning for CAT} \label{subsec:doubleQ}

\textcolor{black}{In classical tabular Q-learning, the action-value function is stored as a Q-matrix whose rows correspond to states and whose columns correspond to actions. Such a representation is infeasible in CAT because the state space is combinatorially large and the state variable $s_t$ is represented by continuous posterior summaries rather than a small finite index. We therefore replace the tabular Q-matrix by a neural-network approximator $Q_w(s,a)$. For a given state $s_t$, the network in Figure \ref{fig:policy-architecture} outputs the full vector of estimated Q-values before selecting the optimal item, which plays the same role as one row of a tabular Q-matrix, but is produced through function approximation rather than table lookup.}

The standard Q-learning algorithm often overestimates Q-values because the same action-value approximation is used both to select the maximizing action and to evaluate it when constructing the temporal-difference target. To address this issue, we adopt the Double Q-learning approach \citep{van_Hasselt_Guez_Silver_2016}, which replaces the single action-value approximation by two Q-networks with distinct roles. \textcolor{black}{The primary network $Q_w$ is updated by gradient descent and is used to select the action with the largest estimated Q-value, while the target network $Q_{\bar w}$ is a delayed copy used only to evaluate that selected action when forming the target. This is the neural-network analogue of maintaining two Q-matrices in Double Q-learning. The separation reduces overestimation bias and leads to more stable training.} Specifically, compared with standard Q-learning, Double Q-learning minimizes the same TD loss in (\ref{eq:td-loss}) but replaces the target in (\ref{eq:pred-q}) with
\[
y(s,a,r,s') \;=\; r + \gamma\, Q_{\bar{w}}\!\big(s',\, \operatorname*{arg\,max}_{a'} Q_{w}(s',a')\big).
\]

Our proposed algorithm is presented in Algorithm \ref{alg:q-learning}. \textcolor{black}{The algorithm trains, entirely offline, the policy network illustrated in Figure \ref{fig:policy-architecture}. Once training is complete, the learned network is deployed during live CAT with fixed weights, and only the posterior state is updated sequentially as new responses are observed.} Importantly, the offline training phase in Algorithm \ref{alg:q-learning} relies only on the item parameters $(\bm{B},\bm{D})$ and simulated examinees drawn from a specified latent factor distribution (e.g., $\bm{\theta}_i \sim \mathcal{N}(0, \mathbbm{I}_K)$), rather than on observed item response data from the item bank. If prior information about the distribution of future online examinees is available, the simulation distribution can be adapted accordingly. In this manuscript, we consistently use a standard multivariate normal distribution to ensure comparability across methods and experiments.
 
Note that for the Q-learning to converge to an optimal policy, it is essential to adopt an $\epsilon$-greedy policy, where the algorithm makes random item selections with probability  $\epsilon$ and gradually decreases $\epsilon$ over the course of training. For simplicity, the state variable $s_t$ in Algorithm \ref{alg:q-learning} is a compact representation of $(\tilde{\bm{\xi}}_t, \bm{\Psi}_t)$ defined in Section \ref{subsec:network-design}. In practice, we can terminate the training when both the rewards and the validation loss stabilize, indicating that the neural network has well approximated the optimal policy.

\RestyleAlgo{ruled}
\SetKwComment{Comment}{/* }{ */}

\begin{algorithm}[hbt!] 
\caption{Double Q-learning Algorithm for Bayesian CAT (offline)}\label{alg:q-learning}
\KwIn{Item Parameters $(\bm{B},\bm{D})$, total Episodes $E$, $\epsilon \in [\underline{\epsilon},\bar{\epsilon}]$ with $\epsilon^{(0)}= \bar{\epsilon}$, buffer $\mathcal{H}$, target network update frequency $T$, batch size $M$.}
\KwOut{Learned Q-network $Q_w$ for item selection.}
Initialize primary Q-network $Q_w$ and target Q-network $Q_{\bar{w}}$, with weights $w_0 = \bar{w}_0$;  \\
\For{$i=1$ \KwTo $E$}{
    Draw $\bm{\theta}_i$ from $\mathcal{N}(0, \mathbbm{I}_{K})$; \\
    Set initial state $s_0 := [\mathcal{N}(0, \mathbbm{I}_{K}), \bm{\Psi}_0]$ and available item set $R_0 = \{1, \cdots, J\}$; \\
    \For{$t=0$ \KwTo $H$}{  
        Select action $a_t$ using $\epsilon$-greedy policy:
        \[
            a_t =
            \begin{cases} 
            \text{Random Selection} & \text{with probability } \epsilon^{(t)}    \\
            \arg\max_{a \in R_t} Q_w(s_t, a) & \text{otherwise.}
            \end{cases}
        \] \\
       Simulate response $y_t \sim \text{Bernoulli}\!\left(\Phi\!\left(\bm{B}_{a_t}'\bm{\theta}_i + D_{a_t}\right)\right)$; \\
        Update available item set $R_{t+1} = R_t \setminus a_t$; \\
         Obtain new state $s_{t+1}$ from $(\tilde{\bm{\xi}}_{t+1}, \bm{\Psi}_{t+1})$; \\ 
        Compute reward $r_t$; \\
        \If{$r_t = 0$}{
            Break and proceed to the next episode; 
        }
        Store transition $H_{i,t} := (s_t, a_t, r_t, R_t, s_{t+1})$ into buffer $\mathcal{H}$; \\
        \If{$|\mathcal{H}| > \text{batch size } M$}{
            Randomly sample $M$ transitions $\{H_{i,t}\}_{j=1}^{M}$ from $\mathcal{H}$; \\
            Compute $a_j = \arg\max_{a \in R_j} Q_w(s_j', a)$; \\
            Compute target:
            \[
                y^{(j)} =
                \begin{cases} 
                -1 & r_{j} = 0    \\
                -1 + \gamma Q_{\bar{w}}(s_j', a_j) & \text{otherwise.}
                \end{cases}
            \]
            Perform batch gradient descent on loss: $\sum_{j=1}^{M} (y^{(j)} - Q_w(s_j, a_j))^2$;
        }
        Decrease $\epsilon^{(t)}$ if $\epsilon^{(t)} \geq \underline{\epsilon}$;
    }
    Update the target network by setting $\bar w = w$ if $i$ is a multiple of $T$;
}
\end{algorithm}

\section{Simulation} \label{sec:experiments}

\subsection{Simulation Design} \label{subsec:sim-design}
To evaluate our approach in a challenging and realistic simulation setting, we randomly generated a $5$-factor, $150$-item \textcolor{black}{factor loading} matrix $\bm{B} \in \mathbbm R^{150 \times 5}$.  Each column of $\bm{B}$ was initialized by randomly permuting $150$ equally spaced values,  with the magnitude constrained to lie within $[0.3, 3]$. We then imposed a lower triangular structure to ensure identifiability. To better reflect practical datasets, where items rarely load on all five factors, each item was required to load on the first factor and on at most two additional factors beyond the first. A visualization of the loading matrix can be found in Appendix \ref{sec:oracle}. Item intercepts were independently drawn from $\text{Unif}(-1.5, 1.5)$.

The goal of the simulation is to accurately estimate the first three factors (\( \mathcal{K} = \{1,2,3\} \)), while accounting for the presence of factors $4$ and $5$. Leveraging our proposed direct sampling approach outlined in Section \ref{sec:acceleration}, we implemented our double Q-learning algorithm alongside all existing methods discussed in Section \ref{subsec:existing_rules}. These include the EAP approach (Equation \ref{eq:kl}), the Max Pos approach (Equation \ref{eq:sp_kl}), the MI approach (Equation \ref{eq:mi_eq}), and the Max Var approach (Equation \ref{eq:var-pred-mean}). We further modified all baseline information-based criteria so that they also target only the prioritized subset of factors $\mathcal{K}=\{1,2,3\}$. Specifically, the EAP, Max Pos, and MI rules are now written as integrals over $\bm{\theta}_{\mathcal{K}}$ rather than all $5$ factors. These adjustments ensure that all competing CAT algorithms are tuned to the same estimation target, making the simulation comparisons fair and directly comparable.

 For performance evaluation, we generated $N=500$ online examinees with latent traits $\bm{\theta}_i \stackrel{\mathrm{i.i.d.}}{\sim} \mathcal{N}(0, \mathbbm{I}_5)$. For each examinee, we administered a $50$-item test using each of the CAT algorithms. We then compared performance in terms of posterior variance reduction, mean-squared error (MSE), termination efficiency, and item exposure rates based on these $500$ adaptive testing sessions. 

For the Q-learning algorithm, we trained the Q-network for $80,000$ episodes, with exploration parameter $\epsilon$ decreasing linearly from $0.99$ to $0.01$ over $700,000$ steps. We set a sufficiently large $H=60$ items, with discount factor $\gamma=0.95$. To verify that the double Q-learning algorithm indeed converges to the optimal policy $\pi^*$, we provide further details on the training dynamics of our deep Q-learning algorithm, illustrating the increase in rewards until convergence in Appendix \ref{subsec:td}.

\subsection{Simulation Results} \label{subsec:sim-results}

\begin{figure}[t]
    \centering
    \includegraphics[width=0.6\textwidth, height=0.4\textwidth]{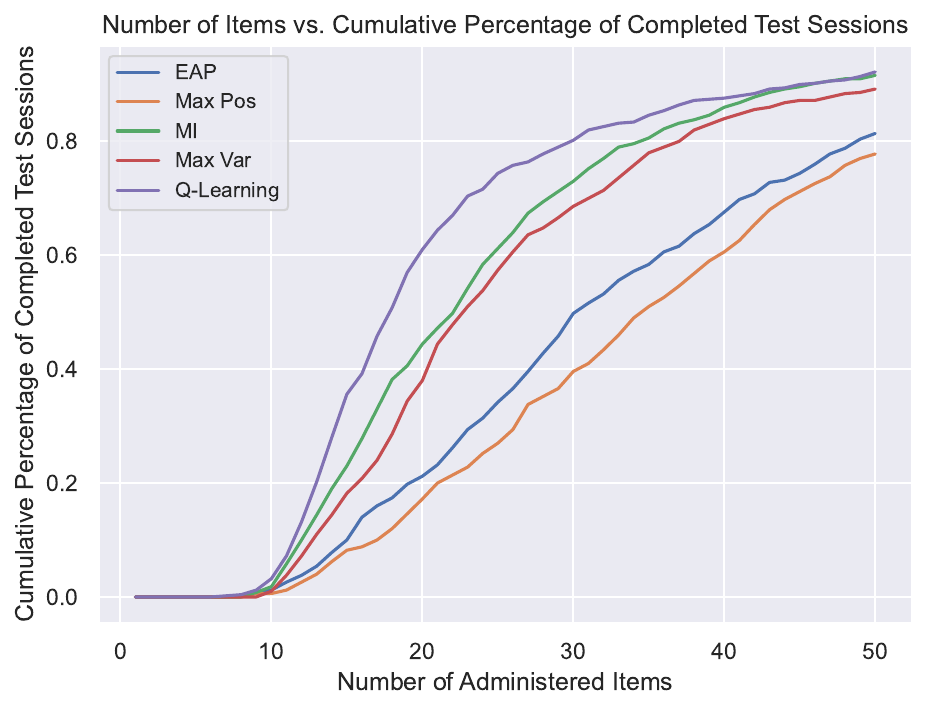}
    \caption{Number of Items Versus Cumulative Percentage of Completed Tests}
    \label{fig:plot1}
\end{figure}

To evaluate termination efficiency, we start with a standard multivariate Gaussian prior with unit marginal variances and terminate the test once the maximum posterior variance across \emph{all three} target factors in $\mathcal{K}$ falls from $1$ to below $\tau^2 < 0.16$. In practice, practitioners may select the threshold $\tau^2$ based on the specific requirements of their application. Based on $500$ simulated adaptive testing sessions, Figure \ref{fig:plot1} illustrates how the percentage of completed test sessions increases as more items are administered. A faster growth rate of the completion percentages indicates faster posterior variance reduction. Notably, the purple Q-learning curve consistently rises more quickly than the others, demonstrating significant testing efficiency gains. As summarized in the second column of Table \ref{tab:mcat_comparison}, Q-learning requires an average of only $21.5$ items to reach the termination criterion for all $3$ targeted factors, outperforming all other methods.

\begin{table}[t]
\centering
\caption{\textbf{Comparison of Win shares (W.S), Termination, and Computation}}
\label{tab:mcat_comparison}
\resizebox{\textwidth}{!}{%
\begin{tabular}{@{}lccccc@{}}
\toprule
\textbf{Algorithm} & \textbf{Avg Termination (items)} & \textbf{W.S dim1} & \textbf{W.S dim2} & \textbf{W.S dim3} & \textbf{Avg Time (s/item)} \\ \midrule
EAP         & 31.7    & 17.0\% & 16.4\% & 13.6\% & \textbf{0.037} \\
Max Pos     & 34.0    & 16.8\% & 16.2\% & 18.6\% & 0.055 \\
MI          & 23.2  & 19.0\% & 21.4\% & 20.2\% & 0.082 \\
Max Var     & 25.7  & 23.6\% & 22.6\% & 23.6\% & 0.038 \\
Q-learning  & \textbf{21.5} & \textbf{23.6}\% & \textbf{23.4\%} & \textbf{24.0\%} & 0.064 \\ \bottomrule
\end{tabular}%
}
\end{table}

To assess early-stage performance, we compute the MSEs after $20$ items for each test session and summarize results using win shares, defined as the percentage of test takers for whom each method achieves the lowest MSE after $20$ items. These win shares are reported across all three latent dimensions in Table \ref{tab:mcat_comparison}. Notably, Q-learning achieves the best win shares across all three factors. This advantage at early stages is consistent with the training objective of Q-learning, which explicitly rewards rapid posterior variance reduction and early test termination, with the average stopping time occurring at around $20$ items.

Table~\ref{tab:mse-factors} reports the evolution of MSEs as a function of test length up to $50$ administered items. Consistent with the win share analysis, Q-learning exhibits competitive performance during the early stages of testing. In particular, at $T=20$ and $T=30$, Q-learning attains the smallest MSEs for Factors $2$ and $3$. As the test length increases, performance differences across methods diminish, and Q-learning achieves MSEs comparable to those of MI and Max Var. This pattern is expected, as the Q-learning policy is trained primarily on short-horizon testing trajectories: since most simulated tests terminate before $30$ items, the algorithm is exposed less frequently to long-horizon simulations during training. Additional visualizations of the MSE trajectories are provided in Appendix~\ref{subsec:matching}.

%%%% MSE TABLE
\begin{table}[t]
\centering
\caption{MSEs between Posterior Mean and Ground Truth for the first three latent factors as a function of test length.}
\label{tab:mse-factors}
\setlength{\tabcolsep}{6pt}
\renewcommand{\arraystretch}{1.1}
\begin{tabular}{llccccc}
\toprule
 & & \multicolumn{5}{c}{Number of items} \\
\cmidrule(lr){3-7}
Factor & Algorithm & 10 & 20 & 30 & 40 & 50 \\
\midrule
\multirow{5}{*}{1}
 & EAP        & 0.175 & 0.123 & 0.090 & 0.073 & 0.062 \\
 & Max Pos     & 0.194 & 0.128 & 0.104 & 0.083 & 0.062 \\
 & MI         & \textbf{0.180} & 0.100 & \textbf{0.066} & \textbf{0.055} & \textbf{0.050} \\
 & Max Var     & 0.186 & \textbf{0.080} & 0.073 & 0.059 & 0.056 \\
 & Q-learning & 0.221 & 0.096 & 0.071 & 0.060 & 0.054 \\
\midrule
\multirow{5}{*}{2}
 & EAP        & 0.362 & 0.214 & 0.164 & 0.125 & 0.105 \\
 & Max Pos     & 0.354 & 0.221 & 0.166 & 0.129 & 0.103 \\
 & MI         & 0.223 & 0.150 & 0.120 & 0.102 & \textbf{0.086} \\
 & Max Var     & \textbf{0.214} & 0.141 & 0.118 & \textbf{0.092} & 0.087 \\
 & Q-learning & 0.282 & \textbf{0.138} & \textbf{0.111} & 0.096 & 0.087 \\
\midrule
\multirow{5}{*}{3}
 & EAP        & 0.317 & 0.226 & 0.152 & 0.112 & 0.085 \\
 & Max Pos     & 0.324 & 0.207 & 0.139 & 0.115 & 0.092 \\
 & MI         & \textbf{0.213} & 0.127 & 0.103 & 0.096 & 0.089 \\
 & Max Var     & 0.234 & 0.136 & 0.110 & 0.099 & 0.096 \\
 & Q-learning & 0.233 & \textbf{0.114} & \textbf{0.098} & \textbf{0.091} & \textbf{0.085} \\
\bottomrule
\end{tabular}
\end{table}

\begin{figure}[t]
    \centering
    \includegraphics[width=0.8\textwidth, height=0.4\textwidth]{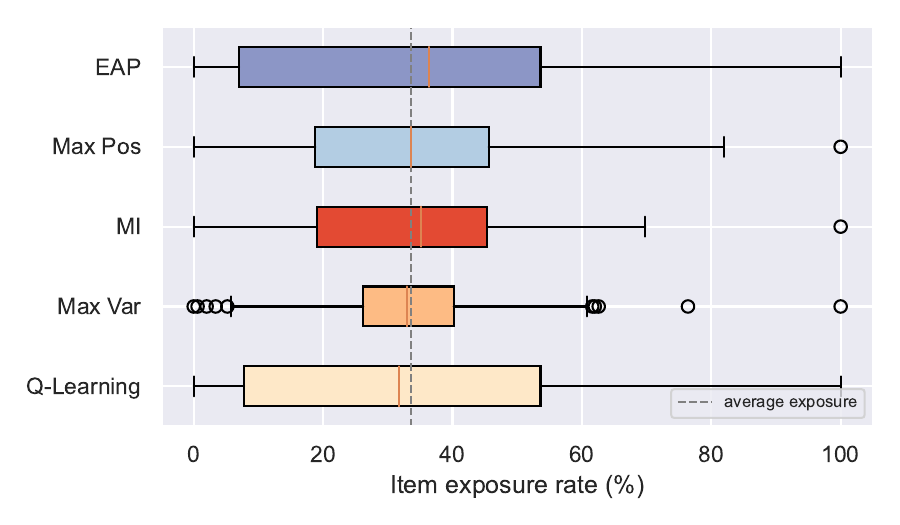}
    \caption{Distributions of Item Exposure Rates}
    \label{fig:exp-rate}
\end{figure}

Figure \ref{fig:exp-rate} summarizes the distribution of item exposure rates across the $150$ item bank for each CAT algorithm, where the exposure rate of an item is defined as the proportion of examinees whose test includes that item at least once. With test length $H=50$ and bank size $J=150$, the expected average exposure is $H/J \approx 0.33$, and all five algorithms yield mean and median exposure rates close to this benchmark. Differences across methods primarily arise in the tails of the exposure distribution. In particular, the EAP and Q-learning policies exhibit slightly higher upper-quartile and maximum exposure rates, indicating a more selective use of highly informative items. Importantly, no method concentrates exposure excessively on a small subset of items, suggesting that all algorithms maintain reasonable item utilization in this simulation setting. For the Q-learning approach, incorporating explicit exposure control mechanisms into the learning objective, such as penalizing repeated use of highly exposed items, is a natural extension but is beyond the scope of the present study.

The last column of Table \ref{tab:mcat_comparison} reports the average online item selection time, highlighting the significant computational advantages of our direct sampling approach. Even the most computationally intensive MI approach requires only $0.082$ seconds per item selection, via our proposed posterior reweighting strategy. Q-learning adds only a single feed-forward pass per selection, keeping online latency under a few milliseconds on standard hardware without GPU acceleration. While training the Q-network offline using Algorithm 1 took approximately $30$ hours for this exercise, this is a one-time investment that can be accelerated if GPUs are available; thereafter, the virtually instantaneous online policy makes the approach highly practical for real‐time CAT deployment.

\section{Cognitive Function Measurements} \label{sec:cat-cog}

\subsection{pCAT-COG Data and Experiment Design} \label{subsec:cat-cog-design}

We revisit the problem of designing a deep CAT system for the pCAT-COG study, as outlined in Section \ref{sec:motivation}. Since item response data for all $N = 730$ examinees across $J = 57$ items are available, we can directly use real item responses during evaluation rather than simulating testing sessions. In this experiment, all $N=730$ examinees are treated as online examinees. All CAT algorithms, including Q-learning, have access only to the estimated item parameters from the pCAT-COG study and do not observe the examinees’ binary item responses beyond those revealed sequentially during adaptive administration.

Given that pCAT-COG is designed to measure global cognitive ability (first column in Figure \ref{fig:cat-cog-loadings}) while accounting for five cognitive subdomains, we specified the Q-learning reward function to reduce posterior variance for the primary factor (\(\mathcal{K} = \{1\}\)). This demonstrates the flexibility of our CAT system, as it can be tailored to the specific cognitive assessment needs. As in Section \ref{sec:experiments}, we also modified all baseline CAT algorithms so that they also target only the primary factor by integrating only over the primary dimension for fair comparison. 

Finally, this real-data experiment also serves as a robustness check for the proposed Q-learning approach. As described in Algorithm \ref{alg:q-learning}, the policy is trained using simulated examinees with latent factors drawn from a standard multivariate normal distribution. In contrast, the empirical factor correlations in pCAT-COG need not follow this distributional assumption. Evaluating Q-learning on real response data therefore provides evidence that the learned policy remains effective when the latent factor structure deviates from the training distribution used in simulation.

\subsection{Results} \label{subsec:cat-cog-results}

The left subplot of Figure \ref{fig:plot4} shows that Q-learning again achieves the fastest test termination compared to other methods. As before, the test is dynamically terminated when the posterior variance drops from $1$ to below \( \tau^2 = 0.16 \), and the cumulative percentage of completed test sessions is computed over all $730$ examinees. The second column of Table \ref{tab:cat-cog-table} shows that Q-learning reaches the desired posterior variance reduction threshold after only an average of $11.2$ items.

For further comparison, we consider adaptively selecting $20$ items for each test taker without dynamic termination. We chose the number $20$ because the Max Pos approach required the largest average of $19.8$ items for termination. The third column of Table \ref{tab:cat-cog-table} shows that the Q-learning approach achieves the highest $0.959$ correlation between the estimated primary factor posterior means and the ground truth after only $20$ items. Additionally, the Q-learning approach attains the highest win shares in estimating the primary dimension across all \( N = 730 \) examinees after $20$ items. As defined in Section \ref{sec:experiments}, win shares denote the proportion of examinees for whom a CAT method attains the smallest mean-squared error. Since all $57$ items are ultimately administered in the real-data study, exposure rates are degenerate under full-length testing. We therefore report item exposure patterns based on the first $30$ administered items in Appendix~\ref{subsec:exposure-cat}.

Additionally, the right subplot of Figure \ref{fig:plot4} highlights the rapid decay of mean squared error (MSE) in estimating the primary dimension with Q-learning. As summarized in Table \ref{tab:realdata-mse-quantile}, Q-learning achieves the smallest MSE at the early stages when $T=10$ and $T=20$. At longer horizons when $T=30$ and $T=40$, the gains in MSE for the Q-learning approach diminish. This behavior is consistent with the simulation findings in Section~\ref{sec:experiments} and reflects the reward specification used for Q-learning, which emphasizes rapid early-stage variance reduction. Because most simulated training trajectories generated in Algorithm~\ref{alg:q-learning} terminate well before 20 items, such performance behavior for Q-learning is expected.

\begin{figure}[t]
    \centering
    \begin{subfigure}{0.48\textwidth}  % First subplot
        \centering
        \includegraphics[width=\linewidth]{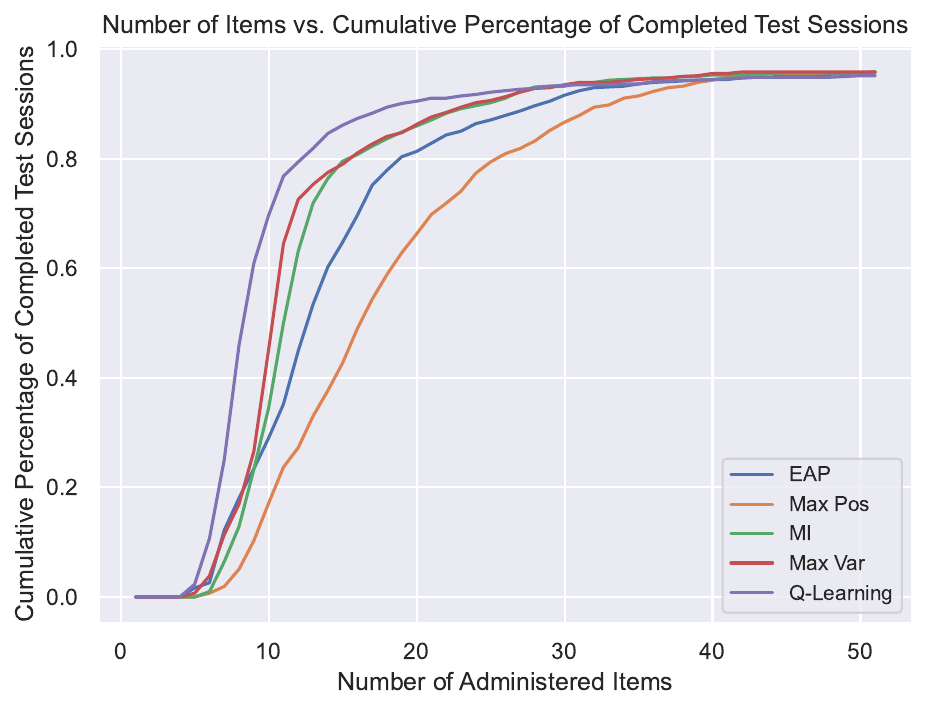}  % First PDF file
        \caption{Test Termination Speed}
    \end{subfigure}
    \hfill  % Space between figures
    \begin{subfigure}{0.5\textwidth}  % Second subplot
        \centering
        \includegraphics[width=\linewidth]{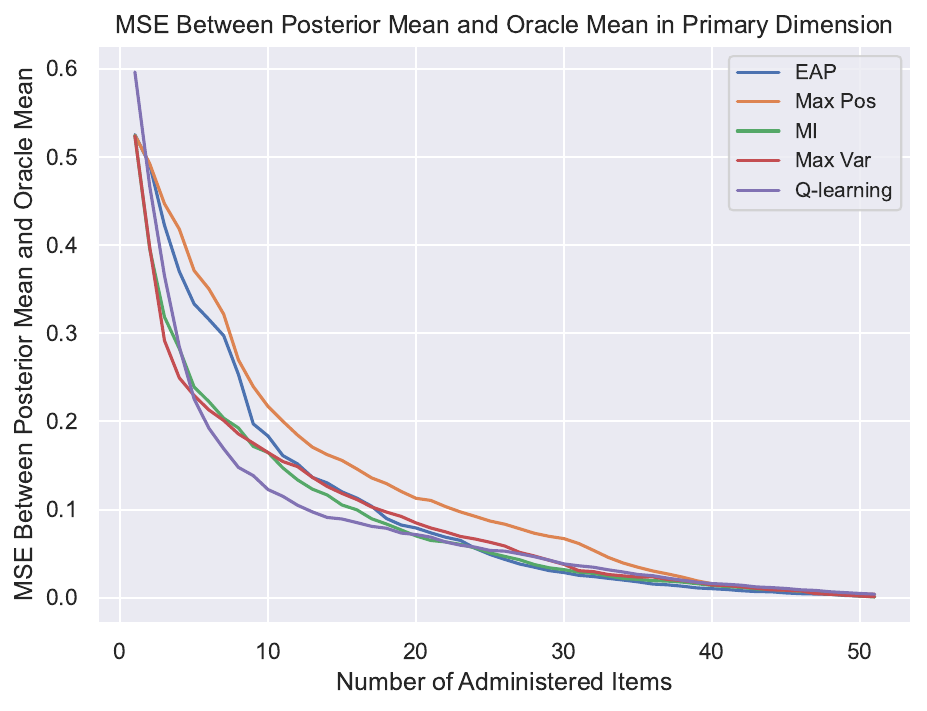}  % Second PDF file
        \caption{Test MSE Accuracy}
    \end{subfigure}
    \caption{pCAT-COG: Primary Factor Posterior Variance Reduction (Left) and Estimation Accuracy (Right)}
    \label{fig:plot4}
\end{figure}

\begin{table}[t]
\centering
\caption{\textbf{Comparison of termination efficiency, primary-factor accuracy, and computation for pCAT-COG.}}
\label{tab:cat-cog-table}
\setlength{\tabcolsep}{8pt}
\renewcommand{\arraystretch}{1.15}
\begin{threeparttable}
\begin{tabular}{lcccc}
\toprule
\textbf{Algorithm} 
& \textbf{Avg Termination (items)} 
& \textbf{Correlation} 
& \textbf{Win shares} 
& \textbf{Avg Time (s/item)} \\
\midrule
EAP        & 15.1 & 0.956 & 17.7\% & \textbf{0.020} \\
Max Pos    & 18.5 & 0.933 & 16.7\% & 0.030 \\
MI         & 12.6 & 0.959 & 21.8\% & 0.038 \\
Max Var    & 13.0 & 0.951 & 20.9\% & 0.021 \\
Q-learning & \textbf{11.2} & \textbf{0.959} & \textbf{22.9}\% & 0.025 \\
\bottomrule
\end{tabular}
\begin{tablenotes}[flushleft]
\footnotesize
\item \textit{Note.} Primary factor correlation and win shares (W.S) are computed after 20 administered items. Win shares denote the proportion of examinees for whom a method attains the smallest mean-squared error for the primary factor at this stage.
\end{tablenotes}
\end{threeparttable}
\end{table}

\begin{table}[t]
\centering
\caption{Mean-squared errors (MSEs) as a function of test length}
\label{tab:realdata-mse-quantile}
\setlength{\tabcolsep}{8pt}
\renewcommand{\arraystretch}{1.1}
\begin{tabular}{lcccc}
\toprule
\textbf{Algorithm} & \textbf{10 items} & \textbf{20 items} & \textbf{30 items} & \textbf{40 items} \\
\midrule
EAP        & 0.183 & 0.079 & \textbf{0.029} & \textbf{0.010} \\
Max Pos    & 0.217 & 0.113 & 0.067 & 0.016 \\
MI         & 0.165 & \textbf{0.070} & 0.032 & 0.014 \\
Max Var    & 0.165 & 0.085 & 0.038 & 0.014 \\
Q-learning & \textbf{0.123} & 0.072 & 0.038 & 0.016 \\
\bottomrule
\end{tabular}
\end{table}

This experiment highlights the effectiveness of our Q-learning approach in high-dimensional cognitive function measurement. Unlike approaches that learn a fixed test design offline \citep{Krantsevich2023}, our framework enables fully dynamic item selection, maximizing the utilization of the item bank by exploring diverse testing trajectories. This termination efficiency is 
particularly valuable in cognitive assessment, where item development is costly, and can help mitigate practice effects and preserve items for future use.

\section{Discussion}

This work advocates for a deep reinforcement learning perspective in the design of multidimensional computerized adaptive testing systems. We make two key contributions to the existing CAT literature: (1) a computational advancement that 
accelerates existing online item selection rules within a flexible Bayesian MIRT framework, and (2) a novel reinforcement learning-based approach that mitigates myopic decision-making and prioritizes the assessment of primary factors of interest.

Theorem \ref{thm:3.2} not only provides an efficient parameterization of the latent factor 
posterior distributions for our proposed deep Q-learning approach, but also improves existing CAT
item selection algorithms, as detailed in Section \ref{sec:acceleration}. Leveraging direct sampling from unified skew-normal distributions, our methodology scales efficiently with a large number of factors and items, achieving near-instantaneous online selection by circumventing MCMC sampling and data augmentation. Additionally, our approach naturally extends to fully Bayesian item selection by accounting for uncertainties in item parameters, which is an essential consideration when the item bank is not well-calibrated as detailed in Appendix \ref{sec:fully-bayesian}. Throughout the testing trajectory, our approach precisely characterizes the evolution 
of posterior distributions at each time step, providing a more robust measurement 
process beyond point estimates with difficult-to-compute standard errors.

Another key contribution of this work is the development of a robust deep double Q-learning algorithm with a customized reward structure that directly minimizes test length. As demonstrated in both simulations and real-data studies, our Q-learning algorithm consistently achieves the 
fastest posterior variance reduction while rapidly decreasing estimation bias. Moreover, its  flexible reward function allows adaptation to different testing objectives, providing a  principled framework for designing customized tests and overcoming the myopic nature of traditional CAT methods.

Our work also offers practical guidance for selecting the appropriate item selection algorithms. Unlike the existing heuristic rules, one limitation of our deep CAT system is the requirement of offline training (Algorithm \ref{alg:q-learning}) before online deployment. For the pCAT-COG study, offline training took approximately $12$ hours on a single GPU as a one-time investment, but no GPU is needed for subsequent online item selection. Even when offline training is undesirable, our framework significantly accelerates existing methods. Experiments show that MI (Equation \ref{eq:mi_eq}) and our Max Var (Equation \ref{eq:var-pred-mean}) approaches often outperform other online item selection rules, echoing the findings presented in \cite{selection_survey}.

A promising direction for future work is to address several limitations of the current Q-learning framework for CAT. First, the learned policy is tied to a specific item bank and must be retrained when the items are substantially modified or replenished, which may limit its immediate applicability in settings with frequent item updates. Second, as a model-free approach, Q-learning can be computationally expensive to train, particularly for long-horizon testing scenarios in which optimal policies depend on extended future trajectories. \textcolor{black}{In addition, although the proposed approach does not rely on tabular exploration of the full state space and instead uses function approximation based on a compact posterior representation, scalability remains an important practical consideration as the state space grows exponentially. Our simulation results suggest that the method works well for moderate 150-item banks, but additional methodological development may be useful for extending the approach to substantially larger banks.} Finally, while the proposed method yields reasonable item exposure patterns in our experiments, additional work is needed to incorporate explicit exposure control mechanisms directly into the learning objective. 

Another exciting avenue is to explore alternative reward structures for our Q-learning algorithm. While the \(0\)-\(1\) reward structure is interpretable and stabilizes Q-network training, it provides sparse feedback, which may limit empirical performance. Alternative designs could incorporate intermediate rewards or increased penalization as more items are administered. An interesting theoretical direction is to characterize when reinforcement learning approaches are particularly advantageous compared to myopic methods. This requires careful assumptions about the reward function and the item bank properties, providing deeper insights into the trade-offs between RL and traditional item selection strategies.

\addtolength{\textheight}{-.3in}%

\spacingset{1.0}
\bibliographystyle{unsrt}
%%%%%%%%\putbib[reference]
%%%%%%%%%%\end{bibunit}
\bibliography{example}

@incollection{kl-is,
  author    = {Mulder, Joris and {van der Linden}, Wim},
  title     = {Multidimensional Adaptive Testing with Kullback--Leibler Information Item Selection},
  booktitle = {Elements of Adaptive Testing},
  year      = {2010},
  pages     = {77--101},
  doi       = {10.1007/978-0-387-85461-8_4},
}

@Article{pekl,
  author  = {Veldkamp, Bernard and {van der Linden}, Wim},
  title   = {Multidimensional adaptive testing with constraints on test content},
  journal = {Psychometrika},
  year    = {2002},
  volume  = {67},
  number  = {4},
  pages   = {575-588},
  month   = {December},
  doi     = {10.1007/BF02295132},
  url     = {https://ideas.repec.org/a/spr/psycho/v67y2002i4p575-588.html}
}

@article{Botev_2016,
	doi = {10.1111/rssb.12162},
  
	url = {https://doi.org/10.1111%2Frssb.12162},
  
	year = 2016,
	month = {feb},
  
	publisher = {Oxford University Press ({OUP})},
  
	volume = {79},
  
	number = {1},
  
	pages = {125--148},
  
	author = {Z. I. Botev},
  
	title = {The Normal Law Under Linear Restrictions: Simulation and Estimation via Minimax Tilting},
  
	journal = {Journal of the Royal Statistical Society Series B: Statistical Methodology}
}

@article{mi_selection,
author = {Alexander Weissman},
title ={Mutual Information Item Selection in Adaptive Classification Testing},

journal = {Educational and Psychological Measurement},
volume = {67},
number = {1},
pages = {41-58},
year = {2007},
doi = {10.1177/0013164406288164},

URL = { 
    
        https://doi.org/10.1177/0013164406288164
    
    

},
eprint = { 
    
        https://doi.org/10.1177/0013164406288164
    
    

}
}

@article{selection_survey, 
author={Wang,Chun and Chang,Hua-hua},
year={2011},
month={07},
title={Item Selection in Multidimensional Computerized Adaptive Testing--Gaining Information from Different Angles},
journal={Psychometrika},
volume={76},
number={3},
pages={363-384},
note={Copyright - The Psychometric Society 2011; Last updated - 2023-12-03},
isbn={00333123},
language={English},
url={http://proxy.uchicago.edu/login?url=https://www.proquest.com/scholarly-journals/item-selection-multidimensional-computerized/docview/873320976/se-2},
}

@article{sir,
 ISSN = {00031305},
 URL = {http://www.jstor.org/stable/2684170},
 author = {A. F. M. Smith and A. E. Gelfand},
 journal = {The American Statistician},
 number = {2},
 pages = {84--88},
 publisher = {[American Statistical Association, Taylor \& Francis, Ltd.]},
 title = {Bayesian Statistics without Tears: A Sampling-Resampling Perspective},
 urldate = {2024-01-31},
 volume = {46},
 year = {1992}
}

@inproceedings{Rnyi1961OnMO,
  title={On Measures of Entropy and Information},
  author={Alfr{\'e}d R{\'e}nyi},
  year={1961},
  url={https://api.semanticscholar.org/CorpusID:123056571}
}

@article{doi:10.1080/01621459.2025.2476786,
author = {Li,Jiguang and Gibbons,Robert and Rockova,Veronika},
title = {Sparse Bayesian Multidimensional Item Response Theory},
journal = {Journal of the American Statistical Association},
year = {2025},
doi = {10.1080/01621459.2025.2476786},

}

@article{chan2009stochastic,
  title={Stochastic Depletion Problems: Effective Myopic Policies for a Class of Dynamic Optimization Problems},
  author={Chan, Carri W. and Farias, Vivek F.},
  journal={Mathematics of Operations Research},
  volume={34},
  number={2},
  pages={333--350},
  year={2009},
  publisher={INFORMS}
}

@inproceedings{Foster2021DeepAD,
  title={Deep Adaptive Design: Amortizing Sequential Bayesian Experimental Design},
  author={Adam Foster and Desi R. Ivanova and Ilyas Malik and Tom Rainforth},
  booktitle={International Conference on Machine Learning},
  year={2021},
  url={https://api.semanticscholar.org/CorpusID:232104961}
}

@inproceedings{10.5555/3294996.3295098,
author = {Zaheer, Manzil and Kottur, Satwik and Ravanbhakhsh, Siamak and P\'{o}czos, Barnab\'{a}s and Salakhutdinov, Ruslan and Smola, Alexander J},
title = {Deep Sets},
year = {2017},
isbn = {9781510860964},
publisher = {Curran Associates Inc.},
address = {Red Hook, NY, USA},
booktitle = {Proceedings of the 31st International Conference on Neural Information Processing Systems},
pages = {3394–3404},
numpages = {11},
location = {Long Beach, California, USA},
series = {NIPS'17}
}

@article{unified_skew_normal,
 ISSN = {03036898, 14679469},
 URL = {http://www.jstor.org/stable/4616942},
 author = {Reinaldo B. Arellano-Valle and Adelchi Azzalini},
 journal = {Scandinavian Journal of Statistics},
 number = {3},
 pages = {561--574},
 publisher = {[Board of the Foundation of the Scandinavian Journal of Statistics, Wiley]},
 title = {On the Unification of Families of Skew-Normal Distributions},
 urldate = {2023-07-10},
 volume = {33},
 year = {2006}
}

@article{Chang1996AGI,
  title={A Global Information Approach to Computerized Adaptive Testing},
  author={Hua-Hua Chang and Zhiliang Ying},
  journal={Applied Psychological Measurement},
  year={1996},
  volume={20},
  pages={213 - 229},
  url={https://api.semanticscholar.org/CorpusID:13626695}
}

@article{Mnih2013PlayingAW,
  title={Playing Atari with Deep Reinforcement Learning},
  author={Volodymyr Mnih and Koray Kavukcuoglu and David Silver and Alex Graves and Ioannis Antonoglou and Daan Wierstra and Martin A. Riedmiller},
  journal={ArXiv},
  year={2013},
  volume={abs/1312.5602},
  url={https://api.semanticscholar.org/CorpusID:15238391}
}

@article{van_Hasselt_Guez_Silver_2016, title={Deep Reinforcement Learning with Double Q-Learning}, volume={30}, url={https://ojs.aaai.org/index.php/AAAI/article/view/10295}, DOI={10.1609/aaai.v30i1.10295}, abstractNote={ &lt;p&gt; The popular Q-learning algorithm is known to overestimate action values under certain conditions. It was not previously known whether, in practice, such overestimations are common, whether they harm performance, and whether they can generally be prevented. In this paper, we answer all these questions affirmatively. In particular, we first show that the recent DQN algorithm, which combines Q-learning with a deep neural network, suffers from substantial overestimations in some games in the Atari 2600 domain. We then show that the idea behind the Double Q-learning algorithm, which was introduced in a tabular setting, can be generalized to work with large-scale function approximation. We propose a specific adaptation to the DQN algorithm and show that the resulting algorithm not only reduces the observed overestimations, as hypothesized, but that this also leads to much better performance on several games. &lt;/p&gt; }, number={1}, journal={Proceedings of the AAAI Conference on Artificial Intelligence}, author={van Hasselt, Hado and Guez, Arthur and Silver, David}, year={2016}, month={Mar.} }

@article{doi:10.3102/10769986221129847,
author = {Xiao Li and Hanchen Xu and Jinming Zhang and Hua-hua Chang},
title ={Deep Reinforcement Learning for Adaptive Learning Systems},

journal = {Journal of Educational and Behavioral Statistics},
volume = {48},
number = {2},
pages = {220-243},
year = {2023},
doi = {10.3102/10769986221129847},

URL = { 
    
        https://doi.org/10.3102/10769986221129847
    
    

},
eprint = { 
    
        https://doi.org/10.3102/10769986221129847
    
    

}
}

@article{Polson01122013,
author = {Nicholas G. Polson and James G. Scott and Jesse Windle},
title = {Bayesian Inference for Logistic Models Using Pólya–Gamma Latent Variables},
journal = {Journal of the American Statistical Association},
volume = {108},
number = {504},
pages = {1339--1349},
year = {2013},
publisher = {ASA Website},
doi = {10.1080/01621459.2013.829001},


URL = { 
    
        https://doi.org/10.1080/01621459.2013.829001
    
    

},
eprint = { 
    
        https://doi.org/10.1080/01621459.2013.829001
    
    

}

}

@article{f3928584-48a1-3ab1-b404-2ec9ffdf51bb,
 ISSN = {01621459, 1537274X},
 URL = {http://www.jstor.org/stable/2290350},
 author = {James H. Albert and Siddhartha Chib},
 journal = {Journal of the American Statistical Association},
 number = {422},
 pages = {669--679},
 publisher = {[American Statistical Association, Taylor & Francis, Ltd.]},
 title = {Bayesian Analysis of Binary and Polychotomous Response Data},
 urldate = {2025-07-05},
 volume = {88},
 year = {1993}
}

@article{Wang_Chang_Boughton_2011, title={Kullback–Leibler Information and Its Applications in Multi-Dimensional Adaptive Testing}, volume={76}, DOI={10.1007/s11336-010-9186-0}, number={1}, journal={Psychometrika}, author={Wang, Chun and Chang, Hua-Hua and Boughton, Keith A.}, year={2011}, pages={13–39}}

@article{a584ed95-55e4-36a9-8d67-ff8dd3f06250,
 ISSN = {10769986, 19351054},
 URL = {http://www.jstor.org/stable/1165370},
 author = {Wim J. van der Linden},
 journal = {Journal of Educational and Behavioral Statistics},
 number = {4},
 pages = {398--412},
 publisher = {[American Educational Research Association, Sage Publications, Inc., American Statistical Association]},
 title = {Multidimensional Adaptive Testing with a Minimum Error-Variance Criterion},
 urldate = {2025-07-05},
 volume = {24},
 year = {1999}
}

@InProceedings{pmlr-v37-schaul15,
  title = 	 {Universal Value Function Approximators},
  author = 	 {Schaul, Tom and Horgan, Daniel and Gregor, Karol and Silver, David},
  booktitle = 	 {Proceedings of the 32nd International Conference on Machine Learning},
  pages = 	 {1312--1320},
  year = 	 {2015},
  editor = 	 {Bach, Francis and Blei, David},
  volume = 	 {37},
  series = 	 {Proceedings of Machine Learning Research},
  address = 	 {Lille, France},
  month = 	 {07--09 Jul},
  publisher =    {PMLR},
  url = 	 {https://proceedings.mlr.press/v37/schaul15.html}
}

@article{Mnih2015,
  author    = {Mnih, Volodymyr and Kavukcuoglu, Koray and Silver, David and Rusu, Andrei A. and Veness, Joel and Bellemare, Marc G. and Graves, Alex and Riedmiller, Martin and Fidjeland, Andreas K. and Ostrovski, Georg and Petersen, Stig and Beattie, Charles and Sadik, Amir and Antonoglou, Ioannis and King, Helen and Kumaran, Dharshan and Wierstra, Daan and Legg, Shane and Hassabis, Demis},
  title     = {Human-level control through deep reinforcement learning},
  journal   = {Nature},
  volume    = {518},
  pages     = {529--533},
  year      = {2015},
  doi       = {10.1038/nature14236},
  url       = {https://doi.org/10.1038/nature14236}
}

@book{10.5555/560669,
author = {Bertsekas, Dimitri P. and Tsitsiklis, John N.},
title = {Neuro-Dynamic Programming},
year = {1996},
isbn = {1886529108},
publisher = {Athena Scientific},
edition = {1st}
}

@book{van2010elements,
  title={Elements of Adaptive Testing},
  author={van der Linden, Wim J. and Glas, Cees A. W.},
  year={2010},
  publisher={Springer},
  address={New York},
  isbn={978-0-387-85461-4},
  doi={10.1007/978-0-387-85461-8}
}

@book{bock2021item,
  title={Item Response Theory},
  author={Bock, R.D. and Gibbons, R.D.},
  isbn={9781119716686},
  lccn={2020055709},
  url={https://books.google.com/books?id=wA\_kzQEACAAJ},
  year={2021},
  publisher={Wiley}
}

@article{gibbons2016computerized,
  author    = {Gibbons, Robert D. and Weiss, David J. and Frank, Ellen and Kupfer, David},
  title     = {Computerized Adaptive Diagnosis and Testing of Mental Health Disorders},
  journal   = {Annual Review of Clinical Psychology},
  year      = {2016},
  volume    = {12},
  pages     = {83-104},
  doi       = {10.1146/annurev-clinpsy-021815-093634},
  pmid      = {26651865},
  note      = {Epub 2015 Nov 20}
}

@article{segall1996multidimensional,
  title     = {Multidimensional Adaptive Testing},
  author    = {Segall, Daniel O.},
  journal   = {Psychometrika},
  volume    = {61},
  number    = {2},
  pages     = {331--354},
  year      = {1996},
  publisher = {Springer},
  doi       = {10.1007/BF02294343},
  url       = {https://doi.org/10.1007/BF02294343}
}

@article{mulder2009multidimensional,
  author    = {Mulder, Joram and van der Linden, Wim J.},
  title     = {Multidimensional Adaptive Testing with Optimal Design Criteria for Item Selection},
  journal   = {Psychometrika},
  year      = {2009},
  month     = {Jun},
  volume    = {74},
  number    = {2},
  pages     = {273--296},
  doi       = {10.1007/s11336-008-9097-5},
  pmid      = {20119511},
  pmcid     = {PMC2813188},
  note      = {Epub 2008 Dec 23}
}

@incollection{segall2000principles,
  author    = {Segall, Daniel O.},
  title     = {Principles of Multidimensional Adaptive Testing},
  booktitle = {Computerized Adaptive Testing: Theory and Practice},
  editor    = {van der Linden, Wim J. and Glas, Cees A. W.},
  publisher = {Kluwer Academic},
  address   = {Boston},
  year      = {2000},
  pages     = {53--73}
}

@article{cat-cog,
author = {Gibbons, Robert and Lauderdale, Diane and Wilson, Robert and Bennett, David and Arar, Tesnim and Gallo, David},
year = {2024},
month = {11},
pages = {},
title = {Adaptive measurement of cognitive function based on multidimensional item response theory},
volume = {10},
journal = {Alzheimer's \& Dementia: Translational Research \& Clinical Interventions},
doi = {10.1002/trc2.70018}
}

@Article{bifactor,
  author={Robert D. Gibbons and Donald Hedeker},
  title={{Full-information item bi-factor analysis}},
  journal={Psychometrika},
  year=1992,
  volume={57},
  number={3},
  pages={423-436},
  month={September},
  doi={10.1007/BF02295430},
  url={https://ideas.repec.org/a/spr/psycho/v57y1992i3p423-436.html}
}

@article{fully_bayesian,
author = {Wim J. van der Linden and Hao Ren},
title ={A Fast and Simple Algorithm for Bayesian Adaptive Testing},

journal = {Journal of Educational and Behavioral Statistics},
volume = {45},
number = {1},
pages = {58-85},
year = {2020},
doi = {10.3102/1076998619858970},

URL = { 
    
        https://doi.org/10.3102/1076998619858970
    
    

},
eprint = { 
    
        https://doi.org/10.3102/1076998619858970
    
    

}
}

@article{Krantsevich2023,
  author    = {Chelsea Krantsevich and P. Richard Hahn and Yi Zheng and Charles Katz},
  title     = {Bayesian decision theory for tree-based adaptive screening tests with an application to youth delinquency},
  journal   = {Annals of Applied Statistics},
  volume    = {17},
  number    = {2},
  pages     = {1038--1063},
  year      = {2023},
  month     = {June},
  doi       = {10.1214/22-AOAS1657}
}

@book{Sutton1998,
  author = {Sutton, Richard S. and Barto, Andrew G.},
  edition = {Second},
  publisher = {The MIT Press},
  title = {Reinforcement Learning: An Introduction},
  url = {http://incompleteideas.net/book/the-book-2nd.html},
  year = {2018 }
}

@article{Chang2015,
  author    = {Chang, Hua-Hua},
  title     = {Psychometrics behind Computerized Adaptive Testing},
  journal   = {Psychometrika},
  year      = {2015},
  volume    = {80},
  number    = {1},
  pages     = {1--20},
  month     = {March},
  doi       = {10.1007/s11336-014-9401-5},
  pmid      = {24499939},
  eprint    = {2014},
}

@article{doi:10.1177/01466219922031338,
author = {Hua-Hua Chang and Zhiliang Ying},
title ={a-Stratified Multistage Computerized Adaptive Testing},

journal = {Applied Psychological Measurement},
volume = {23},
number = {3},
pages = {211-222},
year = {1999},
doi = {10.1177/01466219922031338},

URL = { 
    
        https://doi.org/10.1177/01466219922031338
    
    

},
eprint = { 
    
        https://doi.org/10.1177/01466219922031338
    
    

}
}

@article{Silver2016,
  author    = {Silver, David and Huang, Aja and Maddison, Chris J. and Guez, Arthur and Sifre, Laurent and Driessche, George Van Den and Schrittwieser, Julian and Antonoglou, Ioannis and Panneershelvam, Veda and Lanctot, Marc and Dieleman, Sander and Grewe, Dominik and Nham, John and Kalchbrenner, Nal and Sutskever, Ilya and Lillicrap, Timothy and Leach, Madeleine and Kavukcuoglu, Koray and Graepel, Thore and Hassabis, Demis},
  title     = {Mastering the game of Go with deep neural networks and tree search},
  journal   = {Nature},
  year      = {2016},
  volume    = {529},
  number    = {7587},
  pages     = {484--489},
  doi       = {10.1038/nature16961}
}

@article{Durante_2019,
	doi = {10.1093/biomet/asz034},
  
	url = {https://doi.org/10.1093%2Fbiomet%2Fasz034},
  
	year = 2019,
	month = {aug},
  
	publisher = {Oxford University Press ({OUP})},
  
	volume = {106},
  
	number = {4},
  
	pages = {765--779},
  
	author = {Daniele Durante},
  
	title = {Conjugate Bayes for probit regression via unified skew-normal distributions},
  
	journal = {Biometrika}
}

@article{bad_review,
author = {Kathryn Chaloner and Isabella Verdinelli},
title = {{Bayesian Experimental Design: A Review}},
volume = {10},
journal = {Statistical Science},
number = {3},
publisher = {Institute of Mathematical Statistics},
pages = {273 -- 304},
year = {1995},
doi = {10.1214/ss/1177009939},
URL = {https://doi.org/10.1214/ss/1177009939}
}

@article{Kalashnikov2018QTOptSD,
  title={QT-Opt: Scalable Deep Reinforcement Learning for Vision-Based Robotic Manipulation},
  author={Dmitry Kalashnikov and Alex Irpan and Peter Pastor and Julian Ibarz and Alexander Herzog and Eric Jang and Deirdre Quillen and Ethan Holly and Mrinal Kalakrishnan and Vincent Vanhoucke and Sergey Levine},
  journal={ArXiv},
  year={2018},
  volume={abs/1806.10293},
  url={https://api.semanticscholar.org/CorpusID:49470584}
}

@article{bad,
 ISSN = {13697412, 14679868},
 URL = {http://www.jstor.org/stable/2680683},
 author = {Paola Sebastiani and Henry P. Wynn},
 journal = {Journal of the Royal Statistical Society. Series B (Statistical Methodology)},
 number = {1},
 pages = {145--157},
 publisher = {[Royal Statistical Society, Oxford University Press]},
 title = {Maximum Entropy Sampling and Optimal Bayesian Experimental Design},
 urldate = {2025-02-13},
 volume = {62},
 year = {2000}
}

@article{modern_bad_review,
author = {Tom Rainforth and Adam Foster and Desi R. Ivanova and Freddie Bickford Smith},
title = {{Modern Bayesian Experimental Design}},
volume = {39},
journal = {Statistical Science},
number = {1},
publisher = {Institute of Mathematical Statistics},
pages = {100 -- 114},
year = {2024},
doi = {10.1214/23-STS915},
URL = {https://doi.org/10.1214/23-STS915}
}

@inproceedings{NIPS2016_5ea1649a,
 author = {Lam, Remi and Willcox, Karen and Wolpert, David H.},
 booktitle = {Advances in Neural Information Processing Systems},
 editor = {D. Lee and M. Sugiyama and U. Luxburg and I. Guyon and R. Garnett},
 pages = {},
 publisher = {Curran Associates, Inc.},
 title = {Bayesian Optimization with a Finite Budget: An Approximate Dynamic Programming Approach},
 url = {https://proceedings.neurips.cc/paper_files/paper/2016/file/5ea1649a31336092c05438df996a3e59-Paper.pdf},
 volume = {29},
 year = {2016}
}

@inproceedings{srinivas2010gaussian,
  title={Gaussian process optimization in the bandit setting: No regret and experimental design},
  author={Srinivas, Niranjan and Krause, Andreas and Kakade, Sham M and Seeger, Matthias W},
  booktitle={Proceedings of the 27th International Conference on Machine Learning (ICML)},
  pages={1015--1022},
  year={2010},
  organization={Omnipress}
}

@article{alzheimers2024facts,
  title={2024 Alzheimer's Disease Facts and Figures},
  author={{Alzheimer's Association}},
  journal={Alzheimer's \& Dementia},
  year={2024},
  volume={20},
  number={5},
  doi={10.1002/alz.13809},
  url={https://www.alz.org/getmedia/76e51bb6-c003-4d84-8019-e0779d8c4e8d/alzheimers-facts-and-figures.pdf}
}

@article{10.1001/jamanetworkopen.2018.1726,
    author = {Huang, Alison R. and Strombotne, Kiersten L. and Horner, Elizabeth Mokyr and Lapham, Susan J.},
    title = {Adolescent Cognitive Aptitudes and Later-in-Life Alzheimer Disease and Related Disorders},
    journal = {JAMA Network Open},
    volume = {1},
    number = {5},
    pages = {e181726-e181726},
    year = {2018},
    month = {09},
    issn = {2574-3805},
    doi = {10.1001/jamanetworkopen.2018.1726},
    url = {https://doi.org/10.1001/jamanetworkopen.2018.1726},
    eprint = {https://jamanetwork.com/journals/jamanetworkopen/articlepdf/2701735/huang\_2018\_oi\_180105.pdf},
}

@inproceedings{Gordon1993NovelAT,
  title={Novel approach to nonlinear/non-Gaussian Bayesian state estimation},
  author={Neil J. Gordon and David Salmond and Adrian F. M. Smith},
  year={1993},
  url={https://api.semanticscholar.org/CorpusID:12644877}
}

@article{doi:10.1177/0146621619858674,
author = {Chunxi Tan and Ruijian Han and Rougang Ye and Kani Chen},
title ={Adaptive Learning Recommendation Strategy Based on Deep Q-learning},

journal = {Applied Psychological Measurement},
volume = {44},
number = {4},
pages = {251-266},
year = {2020},
doi = {10.1177/0146621619858674},

URL = { 
    
        https://doi.org/10.1177/0146621619858674
    
    

},
eprint = { 
    
        https://doi.org/10.1177/0146621619858674
    
    

}
}

@book{Goodfellow-et-al-2016,
    title={Deep Learning},
    author={Ian Goodfellow and Yoshua Bengio and Aaron Courville},
    publisher={MIT Press},
    note={\url{http://www.deeplearningbook.org}},
    year={2016}
}

@article{ChangStout1993,
  author    = {Hua-Hua Chang and William Stout},
  title     = {The Asymptotic Posterior Normality of the Latent Trait in an {IRT} Model},
  journal   = {Psychometrika},
  year      = {1993},
  volume    = {58},
  number    = {1},
  pages     = {37--52},
  doi       = {10.1007/BF02294469},
  url       = {https://doi.org/10.1007/BF02294469}
}

@article{BeguinGlas2001,
  author    = {B{\'e}guin, A. A. and Glas, C. A. W.},
  title     = {MCMC estimation and some model-fit analysis of multidimensional IRT models},
  journal   = {Psychometrika},
  year      = {2001},
  volume    = {66},
  number    = {4},
  pages     = {541--561},
  doi       = {10.1007/BF02296195},
  url       = {https://doi.org/10.1007/BF02296195}
}

@article{JiangTemplin2019,
  author    = {Jiang, Z. and Templin, J.},
  title     = {Gibbs Samplers for Logistic Item Response Models via the P{\'o}lya--Gamma Distribution: A Computationally Efficient Data-Augmentation Strategy},
  journal   = {Psychometrika},
  year      = {2019},
  volume    = {84},
  number    = {2},
  pages     = {358--374},
  doi       = {10.1007/s11336-018-9641-x},
  url       = {https://doi.org/10.1007/s11336-018-9641-x}
}

@article{Wang_2015, title={On Latent Trait Estimation in Multidimensional Compensatory Item Response Models}, volume={80}, DOI={10.1007/s11336-013-9399-0}, number={2}, journal={Psychometrika}, author={Wang, Chun}, year={2015}, pages={428–449}}

\clearpage
\bigskip
\spacingset{1.0}
\begin{center}
{\Large{\textbf{Supplementary Material: Deep Computerized Adaptive Testing}}}
\end{center}

\appendix

%\begin{bibunit}  % Start a new bibliography unit for the appendix
\renewcommand{\baselinestretch}{1.3} \normalsize

\section{Mutual Information as Prediction Uncertainties} \label{sec:connection}

We observe that all the heuristic Bayesian item selection methods discussed in this manuscript account for prediction uncertainty. In particular, the mutual information item selection rule can be rewritten as follows: 
\begin{align} \label{eq:mi_connection_to_var}
\arg \max_{j_t \in R_t} I_M (\bm{\theta}, y_{j_t}) &= \arg \max_{j_t \in R_t} \sum_{y_{j_t}=0}^{1} \int_{\bm{\theta}} f(\bm{\theta}, y_{j_t} | \bm{Y}_{1:(t-1)}) \log \frac{f(\bm{\theta}, y_{j_t} | \bm{Y}_{1:(t-1)})}{f(\bm{\theta}| \bm{Y}_{1:(t-1)}) f(y_{j_t} | \bm{Y}_{1:(t-1)})} d\bm{\theta} \notag \\
& = \arg \max_{j_t \in R_t} \sum_{y_{j_t}=0}^{1} \int_{\bm{\theta}} f(y_{j_t}|\bm{\theta}) f(\bm{\theta}|\bm{Y}_{1:(t-1)}) \log \frac{f(\bm{\theta}, y_{j_t} | \bm{Y}_{1:(t-1)})}{f(\bm{\theta}| \bm{Y}_{1:(t-1)}) f(y_{j_t} | \bm{Y}_{1:(t-1)})} d\bm{\theta} \notag \\
& = \arg \max_{j_t \in R_t} \sum_{y_{j_t}=0}^{1} \int_{\bm{\theta}} f(y_{j_t}|\bm{\theta}) f(\bm{\theta}|\bm{Y}_{1:(t-1)}) \log \frac{f(y_{j_t} | \bm{\theta})}{f(y_{j_t}|\bm{Y}_{1:(t-1)})} d\bm{\theta} \notag \\
& = \arg \max_{j_t \in R_t}\int_{\bm{\theta}} [\Phi(\bm{B}_{j_t}'\bm{\theta} + D_{j_t}) \log \frac{\Phi(\bm{B}_{j_t}'\bm{\theta} + D_{j_t})}{c_{j_t}} + \\ \notag
& \hspace*{6em} (1-\Phi(\bm{B}_{j_t}'\bm{\theta} + D_{j_t})) \log \frac{(1-\Phi(\bm{B}_{j_t}'\bm{\theta} + D_{j_t}))}{(1-c_{j_t})}] f(\bm{\theta} | \bm{Y}_{1:(t-1)})  d\bm{\theta}. 
\end{align}
Equation (\ref{eq:mi_connection_to_var}) suggests the mutual information criterion selects the item $j_t$ that has the largest expected KL divergence between the Bernoulli distributions parametrized by $\Phi(\bm{B}_{j_t}'\bm{\theta} + D_{j_t})$ and $c_{j_t}$, weighted by the current posterior $f(\bm{\theta}|\bm{Y}_{1:(t-1)})$. 

Observe that mutual information essentially favors the item $j_t$ that has the largest prediction uncertainties around its prediction mean $c_{j_t}$, quantified by the KL divergence. To compare equation (\ref{eq:mi_connection_to_var}) with equation (\ref{eq:var-pred-mean}), note that we have replace the integral of KL divergence from $c_{j_t}$ with variances $(\Phi(\bm{B}_{j_t}' \bm{\theta} + D_{j_t})-c_{j_t})^2$ in (\ref{eq:var-pred-mean}). This suggests that our proposed approach shares similar theoretical properties with the mutual information criterion, but offers a much simpler formula for online computation, avoiding posterior reweighting and log transformations. Empirically, our proposed item selection rule achieves comparable performance to mutual information but with significantly reduced computational time.

\section{Sequential Sampling in CAT} \label{sec:proof}

\subsection{Proof of Theorem \ref{thm:3.2}} \label{subsec:proof-32}
For any $T \geq 1$, we may write the likelihood of item response after administering $T$ items under the two-parameter probit MIRT model as follows: 
\begin{align*}
    \prod_{t=1}^T \Phi(\bm{B}_{j_t}'\bm{\theta} + D_{j_t})^{y_{j_t}}(1- \Phi(\bm{B}_{j_t}' \bm{\theta} + D_{j_t}))^{1-y_{j_t}} &= \prod_{t=1}^T \Phi \{ (2y_{j_t} -1) (\bm{B}_{j_t}' \bm{\theta} + D_{j_t}) \} \\
     & = \Phi_T\!\left(\bm{C}_1\bm{\theta}+\bm{C}_2;\,\mathbbm{I}_T\right)
\end{align*}
Given the standard multivariate normal prior $\bm{\theta} \sim N(0, \mathbbm{I}_K)$, the posterior can be expressed as follows:
\begin{align*}
    \pi(\bm{\theta}| \bm{Y}_{1:T}, \bm{B}_{1:T}, \bm{D}_{1:T}) & \propto  \phi_K(\bm{\theta}; \bm{0},\mathbbm{I}_K) \Phi_T \{ (\bm{C}_1 \bm{\theta} + \bm{C}_2); \mathbbm{I}_{T} \} \\
    & = \phi_K(\bm{\theta}; \bm{0}, \mathbbm{I}_K) \Phi_T \{ \bm{C}_3^{-1}(\bm{C}_1 \bm{\theta} + \bm{C}_2); \bm{C}_3^{-1} \bm{C}_3^{-1} \} \\
    & = \phi_K(\bm{\theta}; \bm{0}, \mathbbm{I}_K) \Phi_T \{ \bm{C}_3^{-1} \bm{C}_2 + \bm{C}_3^{-1}\bm{C}_1 \bm{\theta} ; \bm{C}_3^{-1} \bm{C}_3^{-1} \}. 
\end{align*}
To identify the exact posterior parameters, we draw back to our attention to the probability kernel of the unified skew-normal distribution as defined in Definition \ref{def:usn}. 
We observe that the posterior density above aligns with the SUN density.
Specifically, if $\boldsymbol{\theta}\sim\operatorname{SUN}_{K,T}(\bm{\mu},\bm{\Omega},\bm{\Delta},\bm{\gamma},\bm{\Gamma})$,
its density is
\[
\phi_{K}\!\big(\boldsymbol{\theta};\bm{\mu},\bm{\Omega}\big)\,
\frac{\Phi_{T}\!\left(\bm{\gamma}+\bm{\Delta}^{\top}\bar{\bm{\Omega}}^{-1}\bm{\Omega}^{-1}(\boldsymbol{\theta}-\bm{\mu})\,;\,
\bm{\Gamma}-\bm{\Delta}^{\top}\bar{\bm{\Omega}}^{-1}\bm{\Delta}\right)}
{\Phi_{T}(\bm{\gamma};\bm{\Gamma})}.
\]
We can interpret this as a product of a $K$-variate Gaussian pdf in $\boldsymbol{\theta}$ and a $T$-variate Gaussian cdf in $\boldsymbol{\theta}$, with normalizing constant $\Phi_{T}(\bm{\gamma};\bm{\Gamma})$, since $\Phi_{T}(\bm{\gamma};\bm{\Gamma})$ is not a function of $\boldsymbol{\theta}$ . But this exactly matches our derived latent-factor posterior $\pi(\boldsymbol{\theta}\mid \bm{Y}_{1:T},\bm{B}_{1:T},\bm{D}_{1:T})$ under the MIRT model. It remains to identify $(\bm{\mu},\bm{\Omega},\bm{\Delta},\bm{\gamma},\bm{\Gamma})$ by matching $\phi_{K}\!\big(\boldsymbol{\theta};\bm{\mu},\bm{\Omega}\big)\Phi_{T}\!\left(\bm{\gamma}+\bm{\Delta}^{\top}\bar{\bm{\Omega}}^{-1}\bm{\Omega}^{-1}(\boldsymbol{\theta}-\bm{\mu})\,;\, \bm{\Gamma}-\bm{\Delta}^{\top}\bar{\bm{\Omega}}^{-1}\bm{\Delta}\right)$ to
$$\phi_{K}\!\big(\boldsymbol{\theta};\mathbf{0},\mathbf{I}_{K}\big)\,
\Phi_{T}\!\left(\bold{C}_{3}^{-1}\bold{C}_{2}+\bold{C}_{3}^{-1}\bold{C}_{1}\boldsymbol{\theta}\,;\,
\bold{C}_{3}^{-1}(\bold{C}_{3}^{-1})^{\top}\right).$$
Matching the pdf part immediately yields $\bm{\mu}_{\text{post}}= 0_{K}$ and $\bm{\Omega}_{\text{post}} = \mathbbm{I}_K$. To match the cdf part, we need to match $\bm{\gamma}+\bm{\Delta}^{\top}\bar{\bm{\Omega}}^{-1}\bm{\Omega}^{-1}(\boldsymbol{\theta}-\bm{\mu})$ to $\bold{C}_{3}^{-1}\bold{C}_{2}+\bold{C}_{3}^{-1}\bold{C}_{1}\boldsymbol{\theta}$, which implies $\bm{\gamma}_{\text{post}} = \bm{C}_3^{-1} \bm{C}_2$. 

It remains to solve the following two linear equations:
\begin{align*}
   & \bm{\Delta}' \bar{\bm{\Omega}}^{-1} \bm{\Omega}^{-1} = \bm{C}_3^{-1} \bm{C}_1 \\
   &\bm{\Gamma} - \bm{\Delta}' \bar{\bm{\Omega}}^{-1} \bm{\Delta} = \bm{C}_3^{-1} \bm{C}_3^{-1}.
\end{align*}

Recall that in definition \ref{def:usn}, we let $\bm{\Omega} = \bm{\omega}\bar{\bm{\Omega}}\bm{\omega}$ be the decomposition of covariance matrix into correlation matrix. Since $\bm{\Omega}_{\text{post}} = \mathbbm{I}_K$, we have $\bar{\bm{\Omega}} = \bm{\Omega} = \mathbbm{I}_K$. 
It follows that solving the first equation yields $\bm{\Delta}_{\text{post}} = \bm{C}_1'  \bm{C}_3^{-1}$. Plugging our solution of $\bm{\Delta}_{\text{post}}$ in the second solution yields $$\bm{\Gamma}_{\text{post}} = \bm{C}_3^{-1} \bm{C}_3^{-1} + \bm{C}_3^{-1} \bm{C}_1 \bm{C}_1' \bm{C}_3^{-1} = \bm{C}_3^{-1}(\bm{C}_1\bm{C}_1' + \mathbbm{I}_T) \bm{C}_3^{-1}.$$

Finally, to show this is indeed a unified skew-normal distribution, we need to show the matrix $\bm{\Omega}^*$ as defined in Definition \ref{def:usn}  is indeed a full-rank correlation matrix. To see this, we may decompose $\bm{\Omega}^*$ as follows:

$$
\begin{bmatrix} \bm{C}_3^{-1}(\bm{C}_1\bm{C}_1' + \mathbbm{I}_T)\bm{C}_3^{-1} & \bm{C}_3^{-1} \bm{C}_1 \\
                 \bm{C}_1'\bm{C}_3^{-1} & \mathbbm{I}_K \\
                  \end{bmatrix}  = \begin{bmatrix} \bm{C}_3^{-1} & 0 \\
                  0 & \mathbbm{I}_K\\
   \end{bmatrix}  \times \begin{bmatrix} \bm{C}_1\bm{C}_1' + \mathbbm{I}_T  & \bm{C}_1  \\
                  \bm{C}_1' & \mathbbm{I}_K\\
   \end{bmatrix} 
   \times \begin{bmatrix} \bm{C}_3^{-1} & 0 \\
                  0 & \mathbbm{I}_K\\
   \end{bmatrix}. 
$$

Observe that this is a decomposition of a correlation matrix, where the middle matrix above is the covariance matrix of a $(T+K)$ dimensional random vector $[\bm{z}_1', \bm{z}_2' ]'$, where $\bm{z}_1= \bm{C}_1 \bm{z}_2 + \bm{\epsilon}$, $\bm{z}_2$ is a $K$-dimensional standard multivariate Gaussian vector with identity covariance matrix, and $\bm{\epsilon}$ is a $T$-dimensional standard multivariate Gaussian vector, independent of $\bm{z}_2$.

\subsection{Sequential Sampling Algorithms} \label{subsec:sequential-sampling}
For any Bayesian item selection algorithms, including our proposed deep CAT framework, it is essential to sample from the latent factor posterior distributions $f(\bm{\theta}| \bm{Y}_{1:T})$ for any $T\geq 1$. Theorem \ref{thm:3.2} provides a direct sampling approach to perform such sequential sampling, hence circumventing the need for MCMC algorithms, which cannot be parallelized and require additional tuning and mixing time.

The key observation is that any unified skew-normal distribution $\bm{\theta} \sim \operatorname{SUN}_{K, T}(\bm{\mu}, \bm{\Omega}, \bm{\Delta}, \bm{\gamma}, \bm{\Gamma})$  has the stochastic representation as $\bm{\theta} \overset{d}{=} \bm{\mu} + \bm{\omega} (\bm{V}_0 + \bm{\Delta} \bm{\Gamma}^{-1} \bm{V}_{1, - \bm{\gamma}}),$
where $\bm{V}_0 \sim N(0, \bar{\bm{\Omega}} - \bm{\Delta} \bm{\Gamma}^{-1} \bm{\Delta}') \in \mathbbm R^{K}$ and $V_{1, -\bm{\gamma}}$ is obtained by component-wise truncation below $-\bm{\gamma}$ of a variate $N(0, \bm{\Gamma}) \in \mathbbm R^{T}$. By plugging in the posterior parameters derived in Theorem \ref{thm:3.2}, we have the representation
$$\bm{\theta} \sim \bm{V}_0 + \bm{C}_1'(\bm{C}_1\bm{C}_1' + \mathbbm{I}_T)^{-1} \bm{C}_3 \bm{V}_{1, -\bm{\gamma}}.$$
Hence we can conduct efficient sequential posterior sampling using the following three steps when $t>1$:

\begin{itemize}
    \item \textbf{Step One}: Sample from multivariate normal distribution $$\bm{V}_0 \sim N\!\left(0, \mathbbm{I}_K-\bm{C}_1'(\bm{C}_1\bm{C}_1'+ \mathbbm{I}_T)^{-1}\bm{C}_1\right).$$
    \item  \textbf{Step Two}: Leveraging the minimax tilting method \cite{Botev_2016}, sample from the zero-mean $T$-variate truncated multivariate normal distribution $\bm{V}_{1, -\bm{\gamma}}$ with covariance matrix $\bm{C}_3^{-1}(\bm{C}_1\bm{C}_1'+ \mathbbm{I}_T)\bm{C}_3^{-1}$ and truncation below -$\bm{C}_3^{-1}\bm{C}_2$.
    \item \textbf{Step Three}: perform linear computation $\bm{V}_0 + \bm{C}_1'(\bm{C}_1\bm{C}_1' + \mathbbm{I}_T)^{-1} \bm{C}_3 V_{1,-\bm{\gamma}}.$
\end{itemize}

Note for $t=1$, the sampling steps above can be modified into their univariate analogue. Since steps one and step two are independent, one can simply perform i.i.d sampling from both $\bm{V}_0$ and $\bm{V}_{1,-\bm{\gamma}}$ in parallel to draw large samples of $\bm{\theta} \sim f(\bm{\theta} | \bm{Y}_{1:T})$ in the probit MIRT model. The $t \times t$ matrix inverse $(\bm{C}_1\bm{C}_1' + \mathbbm{I}_T)$ is easy to compute and only needs to be computed once.

\section{Fully Bayesian Item Selection} \label{sec:fully-bayesian}

\subsection{Generalizing to Fully Bayesian Item Selection} \label{subsec:enable_fully_bayes}

Another key advantage of our direct sampling approach is its ability to enable rapid fully Bayesian item selection. Since Section \ref{sec:problem-formulation}, we have assumed that the item bank is well-calibrated, allowing the factor loading matrix and the intercept to be treated as fixed. According to \cite{fully_bayesian}, this assumption can lead to overly optimistic estimation of the latent traits in practice, and is only reasonable when large calibration datasets are available. A fully Bayesian framework offers a principled way to incorporate this uncertainty by integrating over the joint posterior of item parameters and latent traits. However, performing this integration in online item selection has long been considered intractable, limiting its practical adoption. Notably, \cite{fully_bayesian} introduced an efficient MCMC method, but only for unidimensional item response theory model.

Our direct sampling approach extends fully Bayesian item selection to multidimensional IRT while avoiding the computational bottlenecks of MCMC. To illustrate, we consider the mutual information criterion in (\ref{eq:mi-rewrite}). Fully Bayesian inference involves obtaining posterior samples of item parameters $\{\bm{\xi}^{(m)}\}_{m=1}^M$, where $\bm{\xi}^{(m)} := (\bm{B}^{(m)}, \bm{D}^{(m)})$. The mutual information criterion in Equation (\ref{eq:mi-rewrite}) requires computing the posterior predictive probability:
\begin{equation} \label{eq:bayesian-pred}
    f(y_{j_t} | \bm{Y}_{1:(t-1)}) = \int_{\bm{\xi}} \int_{\bm{\theta}} f(y_{j_t}|\bm{\theta})f(\bm{\theta}| \bm{\xi},\bm{Y}_{1:(t-1)}) d\bm{\theta} d\bm{\xi}.
\end{equation}
Similarly, the KL divergence term must integrate out the nuisance item parameters $\bm{\xi}$: 
\begin{equation} \label{eq:bayesian-kl}
    \int_{\bm{\xi}} \int_{\bm{\theta}} f(\bm{\theta}| \bm{Y}_{1:(t-1)}, y_{j_t} , \bm{\xi}) \log \frac{f(y_{j_t} | \bm{\theta})}{f(y_{j_t}|\bm{Y}_{1:(t-1)})} d\bm{\theta} d\bm{\xi}.
\end{equation}
By leveraging Theorem \ref{thm:3.2}, we can exactly identify the posterior distributions for any arbitrary configuration of $f(\bm{\theta}| \bm{Y}_{1:(t-1)}, y_{j_t},  \bm{\xi}^{(m)})$, hence enabling direct sampling in parallel. It follows Equations (\ref{eq:bayesian-pred}) and (\ref{eq:bayesian-kl}) are computationally feasible via Monte Carlo integration.

In contrast, a standard MCMC approach requires constructing $M$ independent Markov chains with additional data augmentation techniques \citep{f3928584-48a1-3ab1-b404-2ec9ffdf51bb, Polson01122013}, each targeting $f(\bm{\theta} | \bm{Y}_{1:(t-1)}, y_{j_t} , \bm{\xi}^{(m)})$, with sampling only possible after all chains have reached convergence. This introduces a significant computational burden, making MCMC impractical for real-time item selection.

\label{subsec:fully_bayesian_simulation}

We evaluate the effectiveness of our direct sampling approach in enabling fully Bayesian online item selection. Unlike traditional methods that treat item parameters as fixed, the fully Bayesian approach explicitly accounts for their uncertainties, which can be substantial when the item response dataset is small or poorly calibrated. To illustrate, we generate a binary item response dataset $\bm{Y}$ with $N = 500$ examinees and $J = 150$ items under a 3-factor probit MIRT model. The true factor loading matrix $\bm{B}$ has dimensions $J \times 3$, and the intercept vector $\bm{D}$ has $150$ elements.

To integrate item parameter uncertainty into online item selection, we fit a Gibbs sampler to the $N \times J$ item response data. We generated $5,000$ MCMC draws, retaining the last $500$ posterior samples of the item parameters $\bm{\Xi} = \{\bm{\xi}^{(m)}\}_{m=1}^{500}$ after burn-in. As described in Section \ref{subsec:enable_fully_bayes}, the fully Bayesian approach marginalizes over the joint distributions of the nuisance item parameters $\bm{\Xi}$ and latent traits. We then implemented all item selection rules depicted in Sections \ref{subsec:existing_rules} in a fully Bayesian manner.

Figure \ref{fig:plotb1} illustrates the cumulative percentage of completed test sessions across all $500$ simulated examinees as more items are administered. The test was dynamically terminated when the posterior standard deviations across all three factors fell below $0.4$. Consistent with the experimental results in Section \ref{sec:experiments} of the manuscript, the mutual information method and our proposed Max Var method demonstrated superior posterior variance reduction. The second column of Table \ref{tab:tb1} shows that the mutual information method terminated after an average of $18.7$ items, whereas the EAP approach required $24.4$ items.

\begin{figure}[t]
    \centering
    \includegraphics[width=0.8\textwidth, height=0.6\textwidth]{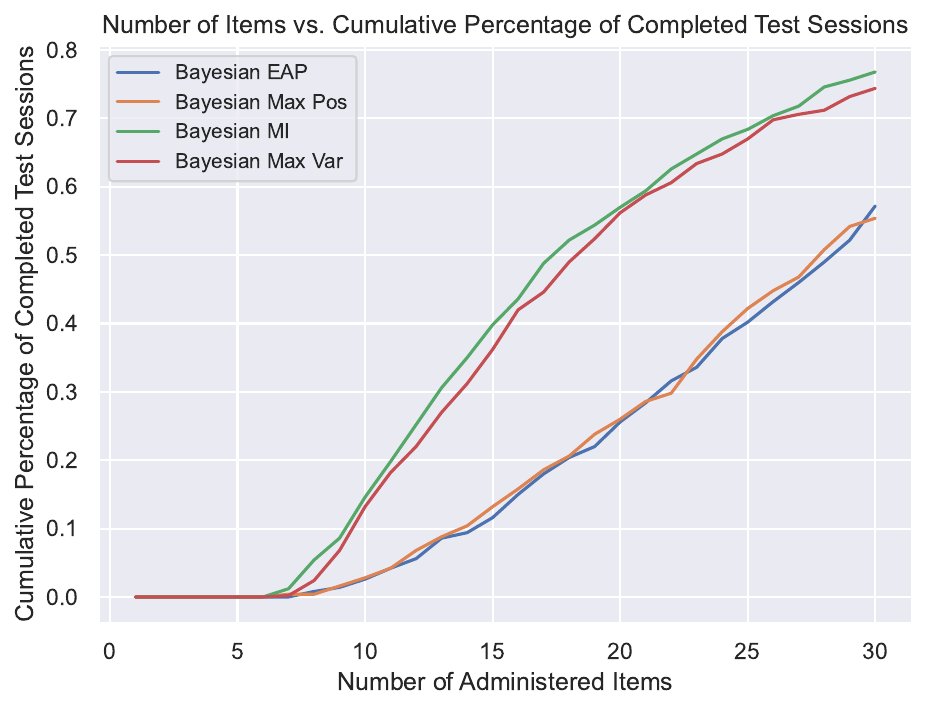}
    \caption{Fully Bayesian: Number of Items Versus Posterior Variance Reduction}
    \label{fig:plotb1}
\end{figure}

\begin{table}[t]
\centering
\caption{\textbf{Comparison of Winshares (W.S), Termination, and Computation}}
\label{tab:tb1}
\resizebox{\textwidth}{!}{%
\begin{tabular}{@{}lccccc@{}}
\toprule
\textbf{Algorithm} & \textbf{Avg Termination (items)} & \textbf{W.S dim0} & \textbf{W.S dim1} & \textbf{W.S dim2} & \textbf{Avg Time (s/item)} \\ \midrule
EAP         & 24.4    & 22.2\% & 19.2\% & 21.8\% & 4.10 \\
Max Pos     & 24.2    & 17.4\% & 18.0\% & 21.0\% & 4.10 \\
MI          & \textbf{18.7}  & \textbf{32.8\%} & 31.0\% & 27.8\% & 4.11 \\
Max Var     & 19.2  & 27.6\% & \textbf{31.8\%} & \textbf{29.4\%} & 4.10  \\
\bottomrule
\end{tabular}%
}
\end{table}

Beyond variance reduction, the mutual information and maximizing prediction variance methods also exhibited superior estimation accuracy, as shown in Figure \ref{fig:plotb2}, which tracks the decline in mean squared error (MSE) between the posterior mean and the oracle posterior mean. Since estimation errors stabilized around $H=30$ items, we compared MSEs between the posterior mean at $H=20$ and the oracle posterior mean, and then summarized the win shares for each item selection rule across all $500$ examinees and all three dimensions in Table \ref{tab:tb1}.

\begin{figure}[t]
    \centering
    \includegraphics[width=1\textwidth, height=0.4\textwidth]{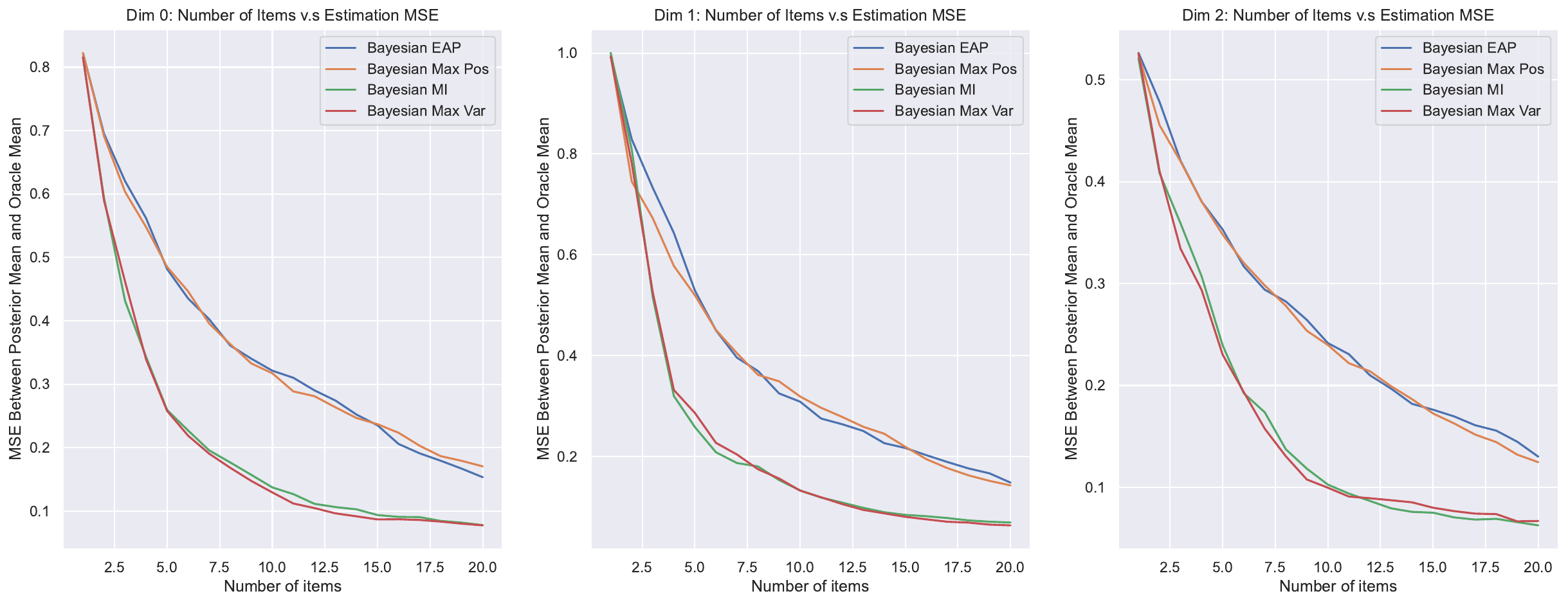}
    \caption{Fully Bayesian: MSE Between Posterior Mean and Oracle Posterior Mean}
    \label{fig:plotb2}
\end{figure}

Finally, we emphasize that even in a fully Bayesian framework, where nuisance parameters are integrated rather than fixed, our approach remains computationally efficient. Across all item selection rules, online selection required an average of just $4.1$ seconds per test session on a personal laptop, highlighting the scalability of our direct sampling method for Bayesian MCAT applications. Readers may notice that $4.1$ seconds is significantly longer than the item selection time reported in the experimentation section, where item parameters are treated as fixed. This discrepancy arises because fully Bayesian MCAT requires sampling from $M=500$ distinct unified-skew normal distributions for each item selection, whereas the standard case involves sampling from a single fixed distribution. This also explains the minimal differences in item selection time across algorithms, as the primary computational bottleneck lies in sampling from these distributions. When more CPUs are available, we can certainly sample from all $M$ distributions in parallel, and hence further improving the $4.1$ seconds benchmark.

\section{Neural Network Design Details} \label{sec:network_details}

Both the primary and target Q–networks share the same modular feed-forward design.  At each time $t$, the posterior parameters $\tilde{\bm{\xi}}_t$ are split into two streams: a permutation-invariant path $\phi_1$ (three $256$-unit ReLU layers) that processes the main state input.  In parallel, a RowwiseNetwork $\phi_2()$ network applies independently to each of the $J$ prediction-quartile feature vectors yields item-wise predictive values, which are then mapped to $256$ dimensional features.  The two $256$-dimensional vectors are then merged by a small combiner MLP, which is then mapped to the final layers with $J$ nodes through the classification network $\rho$.  All linear layers use Xavier initialization; activations are ReLU; the final output is masked to \(-\infty\) for unavailable items.

%Interested readers may check out available code implemented in Pytorch available here: \url{https://github.com/JiguangLi/deep_CAT}.
\begin{table}[ht]
\centering
\caption{Layer‐by‐layer specification of \texttt{OnlineQNetworkV} (and its submodules).}
\label{tab:network-details}
\begin{tabular}{llll}
\toprule
\textbf{Module} & \textbf{Input Dim} & \textbf{Output Dim} & \textbf{Activation} \\
\midrule
\(\phi_1\) (Phi1Network) & \(S\)                     & 256                & ReLU \\
                         & 256                       & 256                & ReLU \\
                         & 256                       & 256                & ReLU \\[0.5em]

\(\phi_2\) (RowwiseNetwork) (for $\bm{\Psi}_t$) & 11                    & 256                  & ReLU \\
                               & 256                   & 256                  & ReLU \\
                               & 256                   & 1                    & Linear
                               \(\times J\) rows \\[0.5em]
                         
\(\phi_2\) (SmallPhi2Network) & \(J\)         &    256               & ReLU \\
                              & 256                 &  256             & ReLU \\
                              & 256                &  256                & Linear \\[0.5em]
\(\rho\)  (Classification)          & 256                & \(J\)                 & ReLU \\
                            & \(J\)                   & \(J\)                 & Linear \\
\bottomrule
\end{tabular}

\medskip
\noindent
\textbf{Notes:}
\begin{itemize}
  \item \(S=\) total scalar‐state dimension; Each new item selection introduced $(K+2)$ dimensional tuple 
  \item The RowwiseNetwork applies its three layers \emph{independently} to each of the \(J\) rows of the \((J\times 11)\) $\bm{\Psi}_t$ matrix. There are $11$ columns, because we include mean, variance, and the $10\%$ quantile increments from $10 \%$ to $90 \%$. Because we process each row independently, the output for the RowwiseNetwork is $J$ dimensional.
  \item All weight matrices are initialized with Xavier uniform; no dropout is used.
  \item Optimizer: Adam with learning rate \(10^{-4}\), \(\beta_1=0.9\), \(\beta_2=0.999\).
\end{itemize}
\end{table}

\section{Discussion on the $0$-$1$ Reward} \label{sec:greedy-proof}

Recall that our choice of reward function is defined in Equation (\ref{eq:reward}). The simple \( 0 \)-\( 1 \) reward structure offers several advantages. It is interpretable as it directly minimizes the number of items required to terminate assessment tasks and prioritizes the main factors of interest. Additionally, it has demonstrated success in related applications, such as in designing deep adaptive learning system \citep{doi:10.3102/10769986221129847}. Since the cumulative reward is constrained to integer values between \( -H \) and \( -1 \), the action-value function can be more reliably approximated by the neural network.

Our double Q-learning algorithm enjoys theoretical convergence guarantees to the optimal policy \citep{van_Hasselt_Guez_Silver_2016} and can accommodate arbitrary reward specifications. It is known that myopic policy cannot outperform the principled optimal policy obtained through RL. However, under certain reward structures, a greedy one-step lookahead (myopic) policy can perform surprisingly well.

In \cite{chan2009stochastic}, it is shown that for stochastic depletion problems, a myopic policy achieves at least  \( 50\% \) of the expected reward of the optimal policy. A similar result can be proven 
for CAT under specific reward structures. Although this can be interpreted as a somewhat negative result for the reinforcement learning approach, we emphasize our choice of the $0$-$1$ reward does not conform to the stochastic depletion problem, and hence is not applicable to the theorem below.

To illustrate this, we extend the state space 
\( \mathcal{S} \) with a time variable, defining \( \overline{\mathcal{S}} = 
\{(f(\bm{\theta} \mid \bm{Y}_{1:T}), t) : t \in [H] \} \). Let $\pi^*(\overline{s})$ represent the optimal policy, and $\pi^m(\overline{s})$ the myopic policy that maximizes the next-step expected reward. Define the mapping $\hat{G}_{y_{a}}: \overline{s} \rightarrow \overline{s'}$ as the posterior update after observing the binary response random variable $y_{a}$, with an incremented time step, where $\overline{s} = (f(\bm{\theta}|\bm{Y}_{1:T}), t)$ and $\overline{s'} = (f(\bm{\theta}|\bm{Y}_{1:T}, y_a), t+1)$. Similarly, define $\tilde{G}_{y_{a}}: \overline{s} \rightarrow \overline{s'}$ as the posterior update after observing $y_{a}$ while keeping the time step fixed, where $\overline{s} = (f(\bm{\theta}|\bm{Y}_{1:T}), t)$ and $\overline{s'} = (f(\bm{\theta}| \bm{Y}_{1:T}, y_{a}), t)$.

To demonstrate the near-optimality of the myopic policy, \cite{chan2009stochastic} 
introduces two key assumptions: ``value function monotonicity" and ``immediate rewards". These assumptions can be adapted to the CAT setting as follows:

\begin{itemize}
    \item \textbf{Immediate Rewards:} for any time stamp $t<H$, suppose the myopic policy $\pi^m$ chooses item $a$, then we must have
    \begin{equation} \label{eq:ir-property}
        V_{\pi^*}(\overline{s}) \leq E[R(\overline{s}, a, \bar{s}' )] + V_{\pi^*} (\tilde{G}_{y_a} (\overline{s})).
    \end{equation}
    Essentially, the property says given a free item (without increasing the time step) chosen by the myopic policy would not deteriorate the optimal rewards.

    \item \textbf{Value Function Monotonicity:} for any time stamp $t<H$, suppose the optimal policy would choose item $a_1$ yet the myopic policy would chooses item $a_2$, then
    \begin{equation} \label{eq:value-mono}
        V_{\pi^*}(\hat{G}_{y_{a_2}} (\tilde{G}_{y_{a_1}}(\overline{s}))) \leq V_{\pi^*} (\hat{G}_{y_{a_2}}(\overline{s})).
    \end{equation}
    Note both sides of inequality refer to the optimal value function evaluated at the posterior distribution at time step $t+1$.
\end{itemize}
In Theorem \ref{thm:myopic-optimality}, we establish that a reinforcement learning approach to CAT cannot significantly outperform the one-step-lookahead greedy policy when the reward function satisfies the two assumptions outlined above. 

\begin{theorem} \label{thm:myopic-optimality}
Consider the finite horizon CAT problem with $H$ steps. Suppose the assumptions of immediate rewards in (\ref{eq:ir-property}) and value function monotonicity in (\ref{eq:value-mono}) hold, we have $V_{\pi^*}(\overline{s}) \leq 2 V_{\pi^m}(\overline{s})$ for all $\overline{s} \in \overline{\mathcal{S}}$.    
\end{theorem}

\begin{proof}
    We proceed by induction. Since the myopic and the optimal policy would agree on time $t = H-1$, the inequality holds trivially for the base case. For the induction step, we assume the claimed inequality holds for all $t'$ such that $H >t' > t$. Consider the time horizon $t < t'$ and let $a_1$ represent the item chosen by the optimal policy and $a_2$ the item chosen by the myopic policy. The inequality holds trivially for $a_1 = a_2$. Hence, we consider the case when $a_1 \neq a_2$ and let $y_{a_1}$ and $y_{a_2}$ represent the binary random variables for the item response. For given extended state $\overline{s}$, we use $\overline{s'}$ to denote the next step after selecting an item. Then we have
    \begin{align}
        V_{\pi^*}(\overline{s} | y_{a_1}, y_{a_2}) &= E[ R(\overline{s}, a_1 ,\overline{s'} ) | y_{a_1}] + V_{\pi^*}(\hat{G}_{y_{a_1}}(\overline{s})) \notag \\
        & \leq E[ R(\overline{s}, a_1 ,\overline{s'} ) | y_{a_1}] + E[ R(\overline{s}, a_2,  \overline{s'} ) | y_{a_2}] + V_{\pi^*}(\tilde{G}_{y_{a_2}}(\hat{G}_{y_{a_1}}(\overline{s}))) \label{eq:asumptiona} \\
        & = E[ R(\overline{s}, a_1 ,\overline{s'} ) | y_{a_1}] + E[ R(\overline{s}, a_2,  \overline{s'} ) | y_{a_2}] + V_{\pi^*}(\hat{G}_{y_{a_1}}(\tilde{G}_{y_{a_2}}(\overline{s}))) \notag \\
        & \leq E[ R(\overline{s}, a_1 ,\overline{s'} ) | y_{a_1}]+ E[ R(\overline{s}, a_2,  \overline{s'} ) | y_{a_2}] + V_{\pi^*}(\hat{G}_{y_{a_2}}(\overline{s}))) \label{eq:asumptionb}\\
        & \leq E[ R(\overline{s}, a_1 ,\overline{s'} ) | y_{a_1}] + E[ R(\overline{s}, a_2,  \overline{s'} ) | y_{a_2}] + 2V_{\pi^m}(\hat{G}_{y_{a_2}}(\overline{s}))). \label{eq:induction}
    \end{align}
Note we applied the immediate rewards assumption in equation (\ref{eq:asumptiona}), the value function monotonicity assumption in equation (\ref{eq:asumptionb}), and the induction hypothesis in equation (\ref{eq:induction}). Now, we may take expectation over $y_{a_1}$ and $y_{a_2}$ on both side of inequality and obtain:
\begin{align*}
  V_{\pi^*}(\overline{s}) & \leq E[ R(\overline{s}, a_1 ,\overline{s'} )] + E[ R(\overline{s}, a_2,  \overline{s'} )] + 2V_{\pi^m}(\hat{G}_{y_{a_2}}(\overline{s}))) \\
  &\leq 2(E[ R(\overline{s}, a_2,  \overline{s'} )] + V_{\pi^m}(\hat{G}_{y_{a_2}}(\overline{s}))) ) = 2 V_{\pi^m}(\overline{s}).
\end{align*}
The last inequality is due to the fact the $\pi^m$ is the myopic policy that maximizes the next step reward.
\end{proof}

Theorem \ref{thm:myopic-optimality} can be seen as a limitation of the reinforcement 
learning approach when the reward structure aligns with the stochastic depletion problem, as it implies that a simple myopic policy provides a 2-approximation to the optimal policy. However, it also highlights the advantage of the simple $0$-$1$ reward function as the value function monotonicity is unlikely to hold, rendering the theorem inapplicable. Specifically, the expected posterior variance of $\hat{G}_{y_{a_2}} (\tilde{G}_{y_{a_1}}(\overline{s})))$ is expected to be smaller than that of $\hat{G}_{y_{a_2}}(\overline{s})$. Consequently, fewer items are expected for the posterior $\hat{G}_{y_{a_2}} (\tilde{G}_{y_{a_1}}(\overline{s})))$ to reach the minimum variance threshold, which implies $V_{\pi^*}(\hat{G}_{y_{a_2}} (\tilde{G}_{y_{a_1}}(\overline{s}))) \geq V_{\pi^*} (\hat{G}_{y_{a_2}}(\overline{s}))$.

%\begin{figure}[t]
%    \centering
%    \includegraphics[width=0.99\textwidth, height=0.55\textwidth]{paper_images/training_dynamics_cat_cog.pdf}
%    \caption{Training Dynamics for the Cognitive Function Measurement Experiments in Section 5.2}
%    \label{fig:plotd2}
%\end{figure}

%\begin{figure}[t]
%    \centering
%    \includegraphics[width=0.99\textwidth, height=0.55\textwidth]{paper_images/training_dynamics_dese.pdf}
%    \caption{Training Dynamics for the DESE Experiment in Section 5.3}
 %   \label{fig:plotd3}
%\end{figure}

\section{More Analysis on Simulation} \label{sec:oracle}

\subsection{Factor Loading Matrix Visualization}

We visualize the factor loading matrix in our simulation exercise in Figure \ref{fig:sim_B}. Note all the factor loadings are bounded between $(-3,3)$ and the matrix is relatively sparse: all the items are loaded on the first factor, but are only allowed to load for at most two more factors. Although identification of the factor loading matrix is not a primary concern in CAT, as it is typically treated as fixed, we purposefully imposed a lower-triangular structure to ensure identifiability and create a more realistic simulation setting.

\begin{figure}[t] 
    \centering
    \includegraphics[width=0.7\textwidth, height=0.7\textwidth]{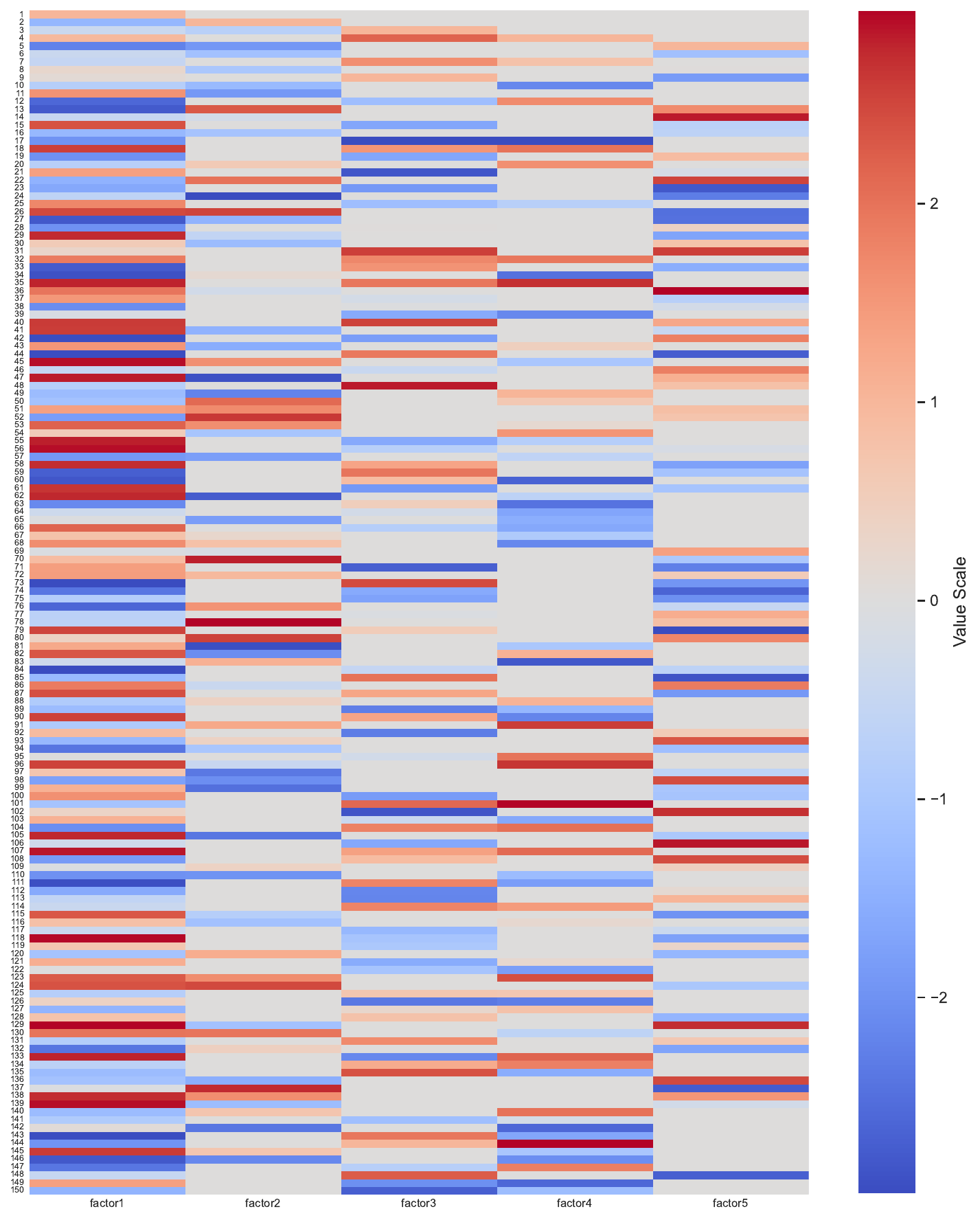}
    \caption{Simulated Factor Loading Matrix}
    \label{fig:sim_B}
\end{figure}

\subsection{Double Q-Learning Training Dynamics} \label{subsec:td}

In Figure \ref{fig:rewards_growth}, we present a metric illustrating the training dynamics of our Q-learning algorithm. The left subplot displays the average episode rewards, computed over every 500 training episodes, showing that rewards began converging around episode 25,000 at approximately \(-23\). The reward is bounded within \([-40,-20]\). Given our definition of the $0-1$ reward, the implication is that a random policy would take $40$ items for the test to converge. After training, it would only take around $20$ items.   

Since we save the primary neural network every 1000 episodes, we select the checkpoint corresponding to the highest average reward for offline deployment in future item selection tasks. In this case, the network at episode 72,000 is chosen, as it achieves the peak average reward of approximately \(-21\). This model selection strategy is consistently applied across all experiments in Section 5.

\begin{figure}[t] 
    \centering
    \includegraphics[width=0.99\textwidth, height=0.55\textwidth]{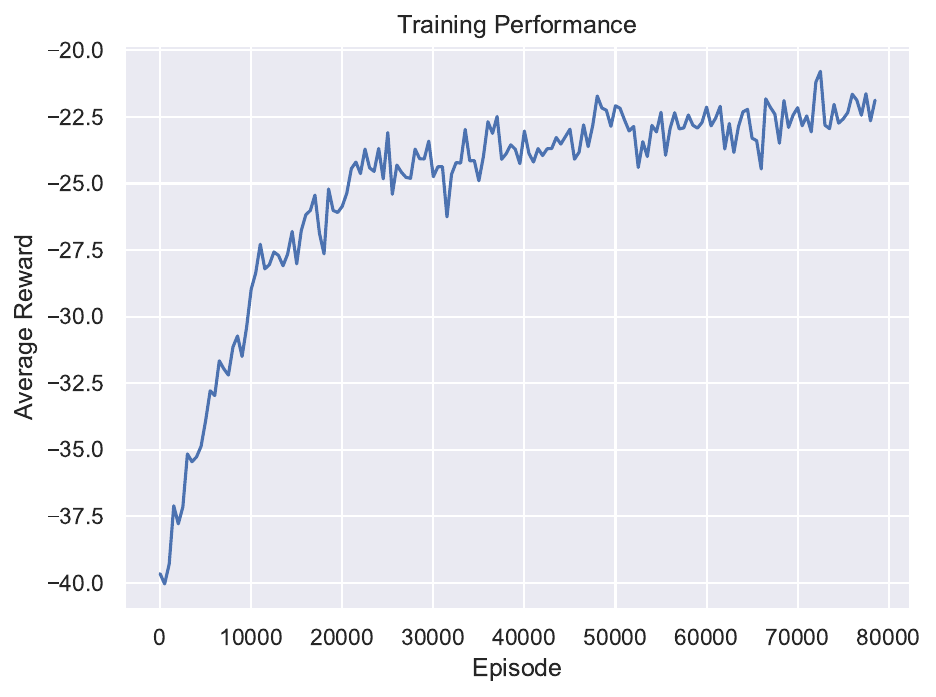}
    \caption{Training Dynamics for the Simulation Experiment in Section 5.1}
    \label{fig:rewards_growth}
\end{figure}

\subsection{Matching the Oracle Posterior} \label{subsec:matching}

Recall the "oracle posterior distribution" of a test taker is defined as the posterior obtained if a given test taker had answered all items in the item bank. To obtain the oracle distributions for each of the $500$ test takers in Section \ref{sec:experiments}, we simulate their item response to all $150$ items in the test bank, and then calculate the posterior distribution using Theorem \ref{thm:3.2}. 

Figures \ref{fig:plotB1}, \ref{fig:plotB2}, and \ref{fig:plotB3} illustrate the decrease of the estimation MSE for the first quartile, median, and  the third quartiles compared to those of the oracle distribution, respectively. The quartiles of each posterior distribution are computed via $1,000$ independent draws. In all these figures, we observe similar superior performances for the MI, Max Var, and the Q-learning methods over the EAP and the Max Pos approaches in estimating the entire oracle distributions. Consistent with the findings in Section \ref{sec:experiments}, the Q-learning approach demonstrates the fastest error reduction rate, especially in the early stages of the test.

\begin{figure}[t]
    \centering
    \includegraphics[width=1.12\textwidth, height=0.55\textwidth]{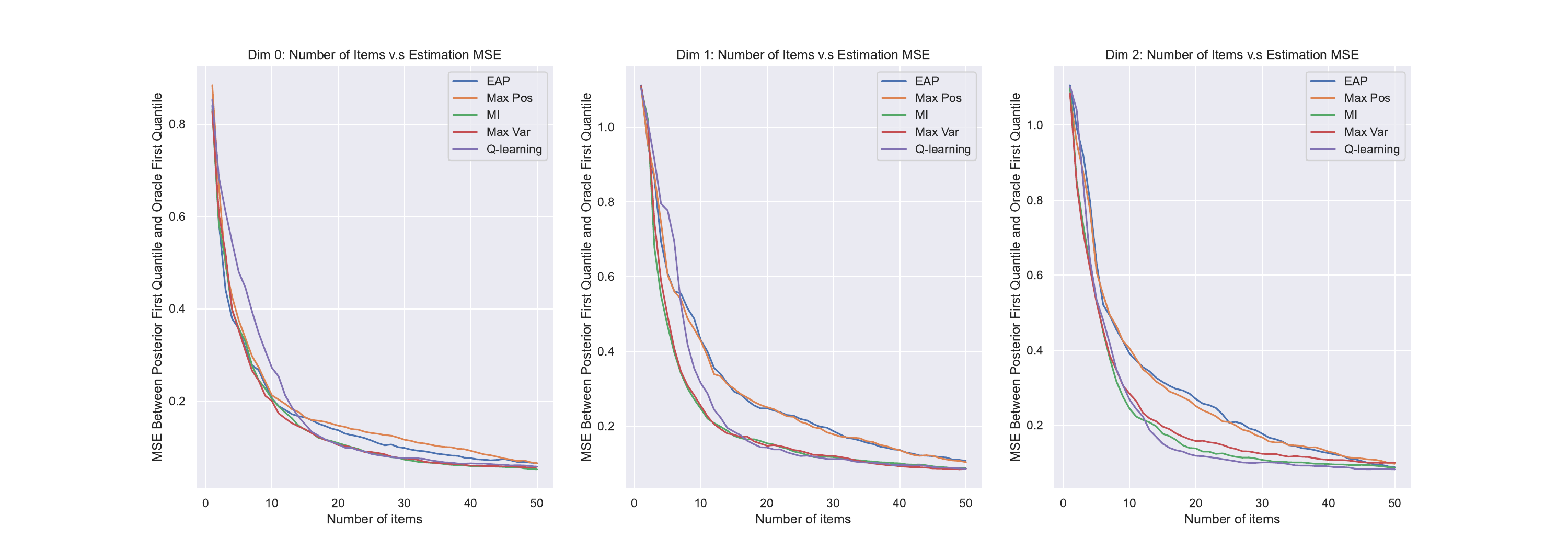}
    \caption{MSE Between Posterior 1st Quantile and Oracle 1st Quantile}
    \label{fig:plotB1}
\end{figure}

\begin{figure}[t]
   \centering
   \includegraphics[width=1\textwidth, height=0.45\textwidth]{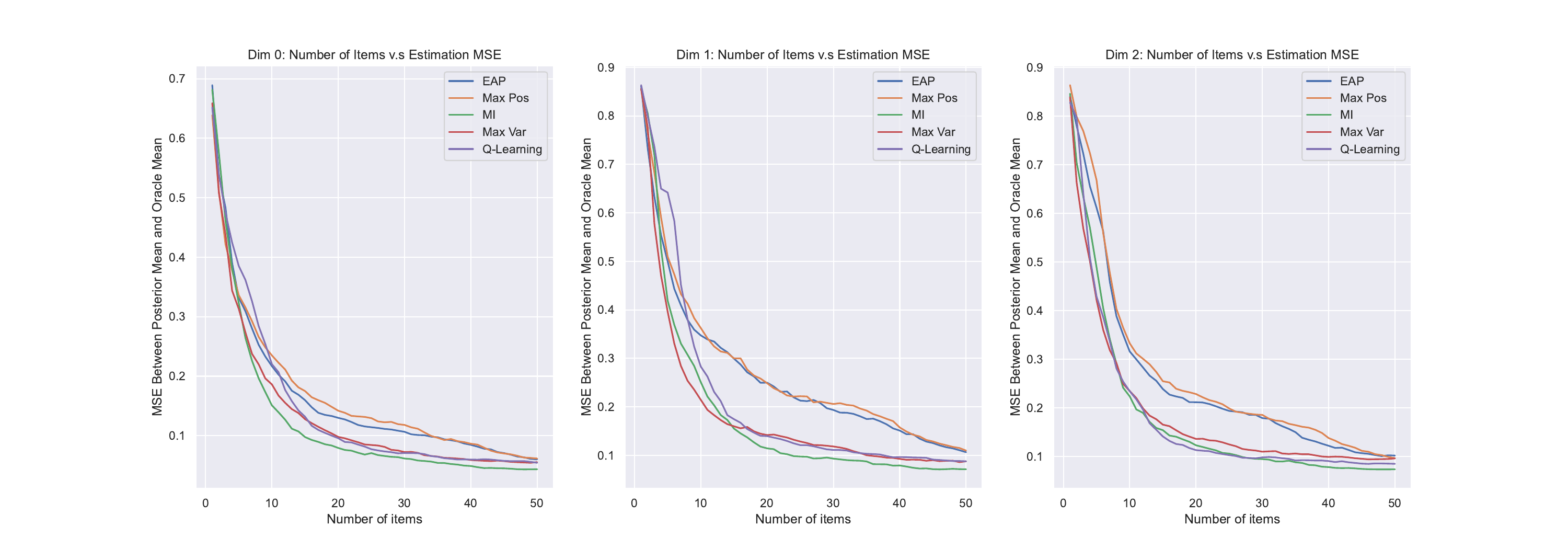}
   \caption{MSE Between Posterior Means and Oracle Posterior Means}
   \label{fig:plot2}
\end{figure}

\begin{figure}[t]
    \centering
    \includegraphics[width=1.12\textwidth, height=0.55\textwidth]{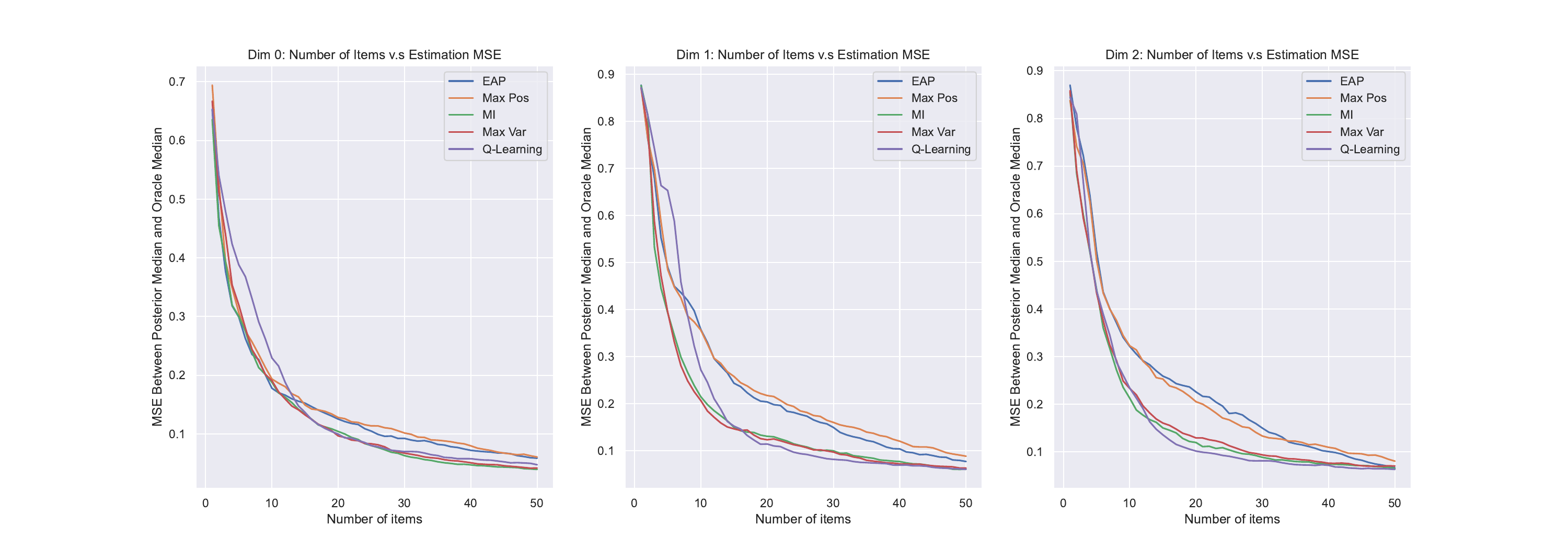}
    \caption{MSE Between Posterior Median and Oracle Median}
    \label{fig:plotB2}
\end{figure}

\begin{figure}[t]
    \centering
    \includegraphics[width=1.12\textwidth, height=0.55\textwidth]{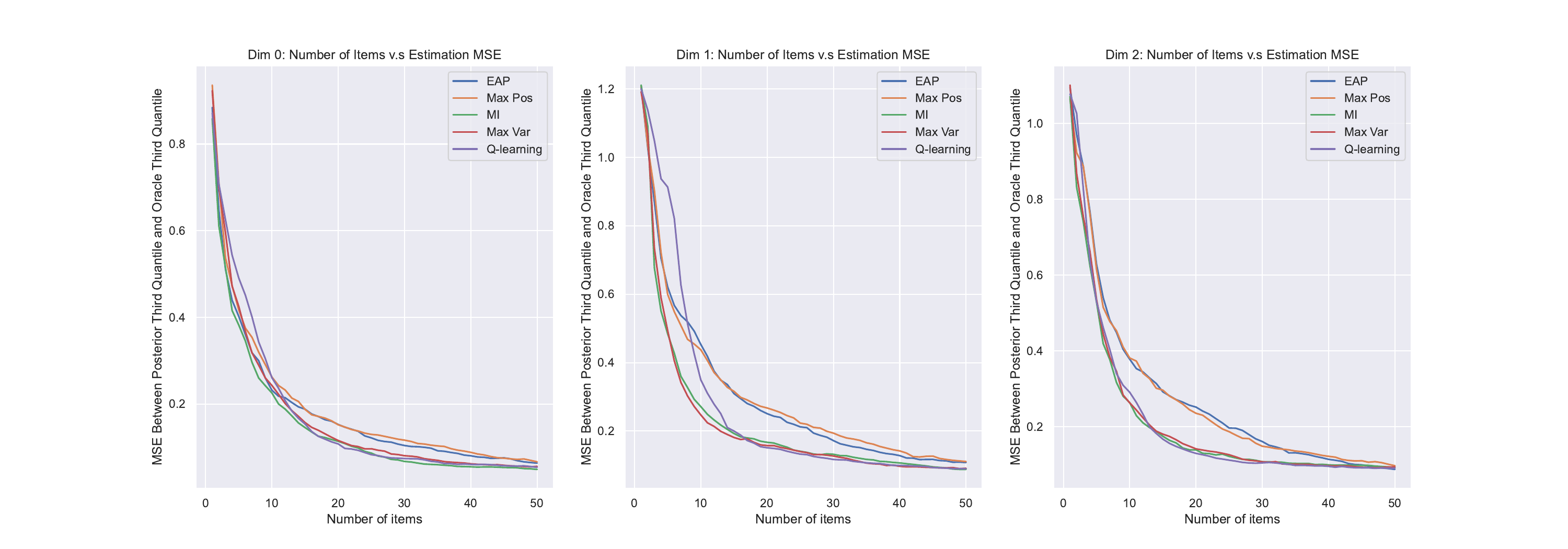}
    \caption{MSE Between Posterior 3rd Quantile and Oracle 3rd Quantile}
    \label{fig:plotB3}
\end{figure}

\section{pCAT-COG Study Additional Visualization} \label{sec:pcat-cog-additional}

\subsection{Exposure rate} \label{subsec:exposure-cat}

In the real-data experiment, we administer the full $57$-item bank in every simulated session to compare the rate of MSE reduction as items accrue. Consequently, item exposure is trivially $100\%$ for all items under all methods, and exposure-rate comparisons are not directly informative in this setting. Therefore, we visualize the distributions of item exposure rate when $T=30$ in Figure \ref{fig:exp-rate-cog}. Since the full item bank has $57$ items, we should expect the exposure rate to be centered around $30/57 \approx 52.6\%$.

\begin{figure}[t]
    \centering
    \includegraphics[width=0.8\textwidth, height=0.4\textwidth]{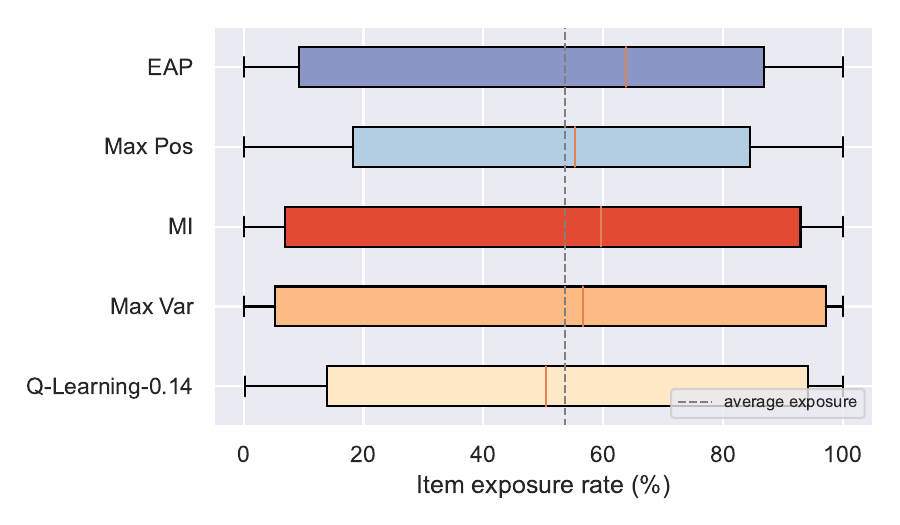}
    \caption{Distributions of Item Exposure Rates at $T=30$: pCAT-COG}
    \label{fig:exp-rate-cog}
\end{figure}

\subsection{Q-learning Training Dynamics}

\section{DESE Educational Assessment} \label{sec:dese_appendix}

We consider an important application in educational assessment, utilizing 2022 Grade 8 student item response data from the Massachusetts Department of Elementary and Secondary Education (DESE). Excluding non-multiple-choice questions, the dataset includes responses to $24$ English items and $34$ math items. For efficiency, we randomly sampled $1,000$ student responses to evaluate each CAT item selection algorithm. 

Given the unknown latent structure, we estimate the factor loading matrix via the PXL-EM algorithm \citep{doi:10.1080/01621459.2025.2476786}, resulting in a $58$ by $4$ factor loading matrix, with visualization and the details of model fitting provided in Appendix \ref{sec:dese_appendix}. When estimating the item parameters, we adopted the idea of dynamic posterior exploration, letting the regularization $\lambda_0$ parameter in the PXL-EM algorithm to gradually increase over the path $\{1, 5, 10, 20, 40, 60, 80, 100\}$. The resulting factor loading matrix $\bm{B}$ is visualized in Figure \ref{fig:plot6}, where dark black regions indicate exact zeros. Notably, all English items load exclusively onto the first factor, representing general ability, while math items can also load onto factors $2$–$4$, reflecting their greater structural complexity. For instance, the final factor ("Math 3") has nonzero loadings only on items $34$, $36$, and $38$, all of which correspond to geometry problems. \footnote{Exam questions are available at: \url{https://www.doe.mass.edu/mcas/2022/release/}.} Interestingly, the authors in \cite{doi:10.1080/01621459.2025.2476786} obtained very similar factor loading matrix estimation for the Grade 10 data as well.

Since all items are heavily loaded on the first primary factor, representing the global latent ability, we focus on the problem of adaptively measuring the first factor while accounting for the underlying multidimensional structure in dimensions $2$-$4$. The left subplot of Figure \ref{fig:plot7} illustrates the cumulative percentage of completed test sessions for the $1,000$ sampled examinees as more items are administered. As shown in the plot, nearly $70\%$ of CAT sessions terminate after only $20$ items under the Q-learning policy. On average over these $1,000$ examinees, Q-learning also required only $20$ items to reach the termination criterion, whereas the second-fastest method (MI) required nearly $26$ items. Beyond its efficiency in variance reduction, Q-learning also achieves a faster decay in MSE when estimating the first factor, reinforcing findings from earlier experiments.
\begin{figure}[t]
    \centering
    \begin{subfigure}{0.48\textwidth}  % First subplot
        \centering
        \includegraphics[width=\linewidth]{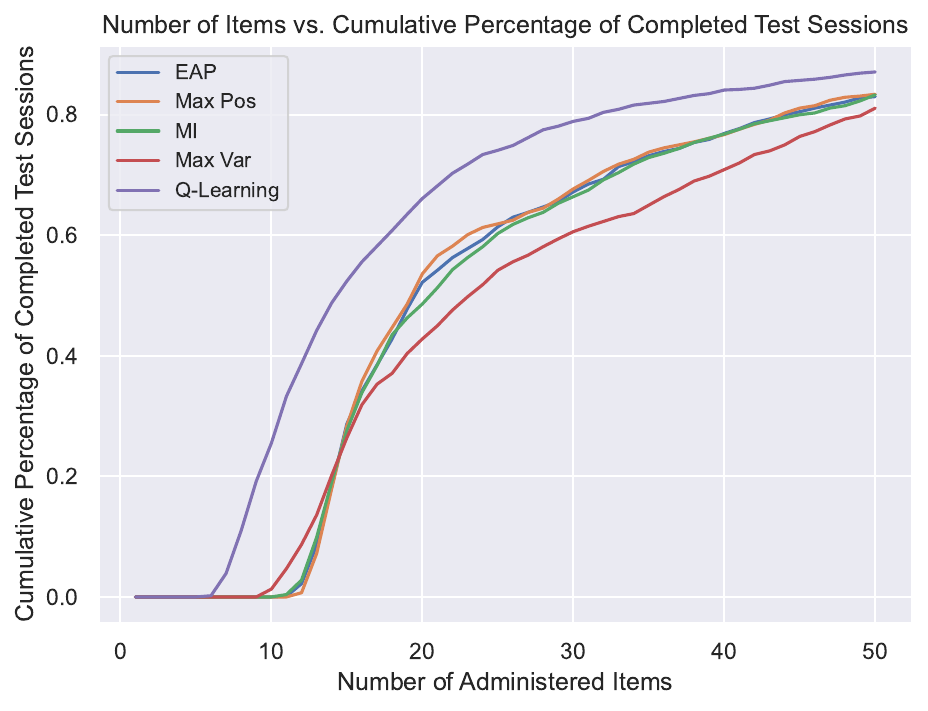}  % First PDF file
        \caption{Test Termination Speed}
    \end{subfigure}
    \hfill  % Space between figures
    \begin{subfigure}{0.5\textwidth}  % Second subplot
        \centering
        \includegraphics[width=\linewidth]{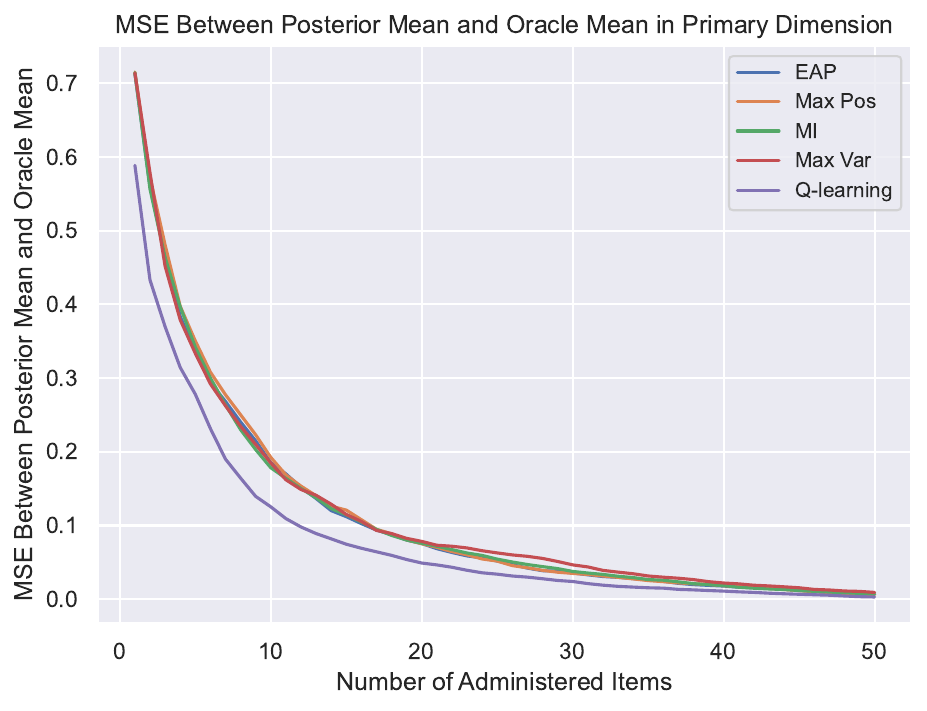}  % Second PDF file
        \caption{Test MSE Accuracy}
    \end{subfigure}
    \caption{Educational Assessment: Main Factor Posterior Variance Reduction (Left) and Estimation Accuracy (Right)}
    \label{fig:plot7}
\end{figure}

These results highlight the potential of Q-learning for educational assessments, demonstrating its ability to adapt to various testing environments and to accelerate testing while maintaining accurate estimation of students' multivariate latent traits.

\begin{figure}[t]
    \centering
    \includegraphics[width=0.5\textwidth, height=0.5\textwidth]{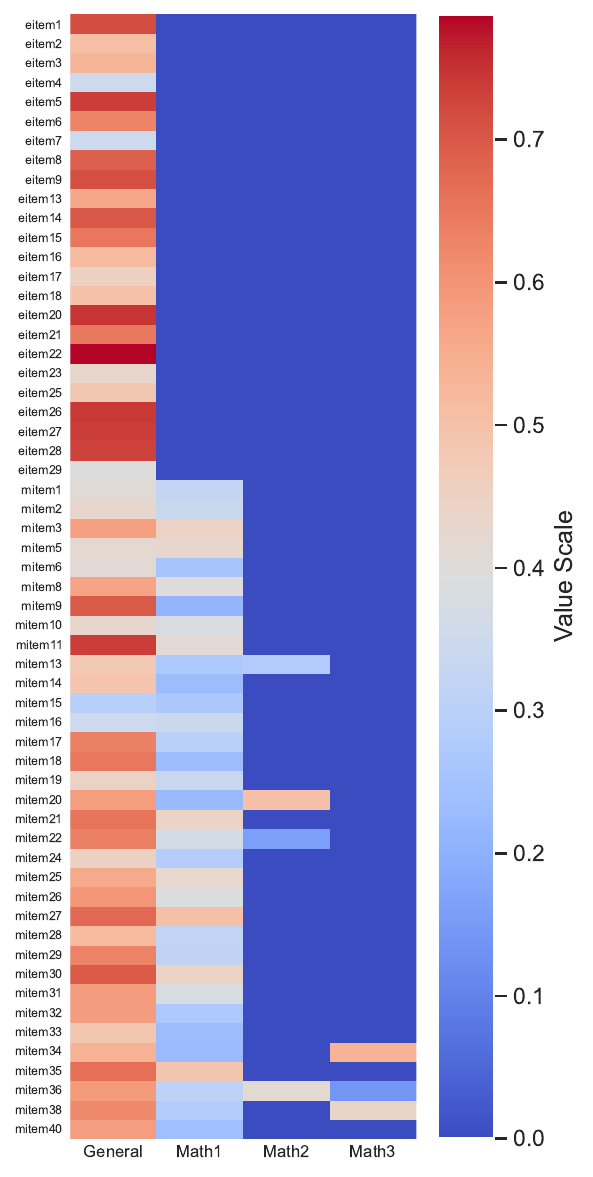}
    \caption{Estimated Factor Loadings for DESE Data}
    \label{fig:plot6}
\end{figure}
%\putbib[reference_appendix]
%\end{bibunit}  % End the bibliography unit for the appendix

\end{document}